\documentclass[12pt]{amsart}

\newtheorem{thm}{Theorem}[section]

\newtheorem{prop}[thm]{Proposition}
\newtheorem{lem}[thm]{Lemma}

\newtheorem{cor}[thm]{Corollary}
\newtheorem{propty}{Property}

\newtheorem{claim}{Claim}
\newtheorem*{reduction}{Reduction}

\usepackage{pdfpages}
\usepackage{parskip}
\usepackage{graphicx}
\usepackage{amsmath, amssymb, amsthm,fullpage,lineno,amsaddr}
\usepackage[toc,page]{appendix}
\usepackage{booktabs}
\usepackage[toc]{appendix}
\usepackage[section]{placeins}

\begin{document}

\title{Oriented Colourings of Graphs with Maximum Degree Three and Four}

\maketitle

\begin{center}
	\large{Christopher Duffy$^{a}$, Gary MacGillivray$^{b}$, \'{E}ric Sopena$^c$}
\end{center}

\begin{center}
	$^{a}$Department of Mathematics and Statistics,  University of Saskatchewan, CANADA.\\
	$^{b}$Department of Mathematics and Statistics, University of Victoria, CANADA.\\
	$^{c}$ Univ. Bordeaux, Bordeaux INP, CNRS, LaBRI, UMR5800, F-33400 Talence, FRANCE.
\end{center}

\begin{abstract}
We show that any orientation of a graph with maximum degree three has an oriented 9-colouring, and that any orientation of a graph with maximum degree four has an oriented 69-colouring.  These results improve the best known upper bounds of 11 and 80, respectively.
\end{abstract}

\section{Introduction}
Recall that an \emph{oriented graph} is a digraph $G$ obtained from a simple, undirected graph $H$ by assigning to each edge one of its two possible \emph{orientations}.  
We refer to $G$ as \emph{an orientation of $H$}.
We refer to
$H$ as the \emph{underlying graph of $G$}, and denote it by $U(G)$.

An \emph{oriented $k$-colouring} of an oriented graph $G$ is a function
$c: V(G) \to \{0, 1, 2, \ldots, $ $k-1\}$ such that 
\begin{enumerate}
	\item $c(u) \neq c(v)$ for all $uv \in E(G)$, and
	\item for all $uv, xy \in E(G)$, if $c(u) = c(y)$, then $c(v) \neq c(x)$.
\end{enumerate}
Condition (1) requires that adjacent vertices are assigned different colours.
Condition (2) can be viewed as the requirement that the colour assignment takes the orientation 
into account: if there is an edge from a vertex of colour $i$ to a vertex of colour $j$,
then there is no edge from a vertex of colour $j$ to a vertex of colour $i$.
The \emph{oriented chromatic number} of the oriented graph $G$, denoted $\chi_o(G)$,  
is the smallest integer $k$ such that 
$G$ has an oriented $k$-colouring.
The survey \cite{SO15} gives a good overview of results and open problems in oriented
colouring and related areas.

It is well-known that the oriented chromatic number of $G$ can differ dramatically from the chromatic number of $U(G)$.
For example, for $n \geq 1$ consider the oriented complete bipartite graph $G$ with vertex set 
$V(G) = \{x_1, x_2, \ldots, x_n\} \cup \{y_1, y_2, \ldots, y_n\}$, and edge set
$E(G) = \{x_i y_j: 1 \leq i \leq j \leq n\} \cup \{y_j x_i: 1 \leq j < i \leq n\}$.
It is easy to see that any two vertices of $G$ are joined by a directed path of length at most two.
By condition (2) in the definition of an oriented $k$-colouring, the ends of a directed path of length 2 (a \emph{$2$-dipath})
must be assigned different colours.
It follows that $\chi_o(G) = 2n$, whereas $\chi(U(G)) = 2$.

For an integer $k \geq 1$, the problem of deciding whether an oriented graph $G$ has an oriented $k$-colouring
is Polynomial if $k \leq 3$, and NP-complete if $k \geq 4$ \cite{KM04}.
On the other hand, since the existence of an oriented $k$-colouring is expressible in monadic second-order logic,
it follows from Courcelle's Theorem that for each positive integer $t$ there is a polynomial-time algorithm to compute 
the oriented chromatic number of orientations of graphs with treewidth at most $t$.

Upper bounds on the oriented chromatic number are known for orientations of graphs 
belonging to many graph families, for
example partial $t$-trees \cite{SO97}, planar graphs \cite{RASO94}, Halin graphs \cite{DS14},
outerplanar graphs \cite{PS06}, hypercubes \cite{Wo07}, and grids \cite{FRR03}.

A Brooks' Theorem for the oriented chromatic number was proved by Sopena using constructive methods
\cite{SO97},
and then improved by Kostochka, Sopena and Zhu using an argument that involves the
probabilistic method  \cite{KSZ97}.  
In particular, if $G$ is an oriented graph for which the underlying graph has maximum degree
$\Delta$, then $\chi_o(G) \leq \Delta^22^{\Delta+1}$.
Further, for each $\Delta > 1$ there exists an orientation of a simple graph with maximum 
degree $\Delta$ and oriented chromatic number at least $2^{\Delta / 2}$   \cite{KSZ97}.

Suppose $G$ is an orientation of a graph with maximum degree $3$.
Sopena proved that $\chi_o(G) \leq 16$ and conjectured that any such connected graph has
an oriented colouring with at most $7$ colours \cite{SO97}.
If this conjecture is true, then the bound is best possible.
The upper bound has been subsequently reduced to $11$ \cite{SV96}.
The main result of this paper reduces it to $9$.

Now suppose $G$ is an orientation of a graph with maximum degree 4.
The bound $\chi_o(G) \leq 80$ can be obtained by
combining the theorem of Raspaud and Sopena 
that any orientation of a graph with acyclic chromatic number $k$
has oriented chromatic number at most $k 2^{k-1}$ with 
the upper bound of 5 on the acyclic chromatic number of a graph with maximum degree 4~\cite{Gr83}.
We use similar ideas as in the proof of the main result to show 
that any such oriented graph has oriented chromatic number at most 69.

The remainder of the paper is organized as follows.  
Background material and relevant definitions are reviewed in the next section.
Some preliminary results are then presented in Section 3.
The fourth section is devoted to the main result and its proof. 
In Section 5, ideas from the proof of the main result are used 
to improve the bound on the oriented chromatic number of orientations of graphs with maximum degree 4.
The paper concludes with some discussion and suggestions for future research.

\section{Background and Terminology}
\label{prelim}

We first review some terminology about directed graphs.

Let $G$ be an oriented graph, and $x \in V(G)$.
The \emph{out-neighbourhood} of $x$ is $N^+(x) = \{y: xy \in E(G)\}$,
and the \emph{in-neighbourhood} of $x$ is $N^-(x) = \{y: yx \in E(G)\}$.
The \emph{out-degree of $x$} is $|N^+(x)|$ and the \emph{in-degree of $x$} is $|N^-(x)|$.
The \emph{degree} of $x$ is $|N^+(x) \cup N^-(x)|$,
i.e.,  its degree in $U(G)$.
For this reason, we use $\Delta$ to denote the maximum degree of a vertex of $G$.

A graph or oriented graph is called \emph{subcubic} if $\Delta \leq 3$ and \emph{subquartic} if $\Delta \leq 4$.
It is \emph{properly subcubic} if $\Delta \leq 3$ and there is a vertex of degree at most 2, i.e., if it is a proper subgraph of a cubic graph.
We define the term \emph{properly subquartic} analogously.

The  vertex $x$  is a \emph{source} if $N^-(x) = \emptyset$, 
and a \emph{sink} if $N^+(x) = \emptyset$.
A \emph{universal source} is a source such that $N^+(x) = V(G) \setminus \{x\}$,
and a \emph{universal sink} is a sink such that $N^-(x) = V(G) \setminus \{x\}$.

The \emph{distance} from $x$ to $y$ is the smallest length of a directed path from $x$ to $y$, or infinity if no 
such path exists. 
The \emph{weak distance between $x$ and $y$} is 
the minimum of the distance from $x$ to $y$ and the distance from $y$ to $x$. This parameter is $\infty$ if 
there is neither a directed path from $x$ to $y$  nor a directed path from $y$ to $x$.
The \emph{weak diameter} of an oriented graph $G$ is the maximum of the weak distance between any two distinct vertices of $G$.

The \emph{converse} of an oriented graph $G$ is the oriented graph obtained by reversing the orientation of each of its arcs. 

We next review some background information about homomorphisms, tournaments and their relationship to oriented colourings.

Let $G$ and $H$ be oriented graphs.

A \emph{homomorphism of $G$ to $H$} is a function $\phi: V(G) \to V(H)$ so that if $uv \in E(G)$, then ${\phi(u)\phi(v) \in E(H)}$. 
If there exists a homomorphism $\phi$ of $G$ to $H$, then we write $\phi: G \to H$ and we say that $\phi$ is a $H$-colouring of $G$.
When the name of the function $\phi$ is clear from the context, or not important, we write $G \to H$.
A homomorphism $\phi: G \to G$ is called an \emph{automorphism} when $\phi$ is a bijection.
We say $G$ is \emph{vertex transitive} when for all $y,z \in V(G)$ there exists an automorphism $\beta$ such that $\beta(y) = z$.
We say $G$ is \emph{arc transitive} when for all $uv, wx \in E(G)$ there exists an automorphism $\rho$ such that $\rho(u) = w$ and $\rho(v) = x$.

An oriented $k$-colouring of an oriented graph $G$ can equivalently be defined 
as a homomorphism to some oriented graph $H$ on $k$ vertices. 
Since  $G \to H$ implies that $G \to H^\prime$ whenever $H$ is a subgraph of $H^\prime$,
the oriented graph $H$ in the previous statement can be assumed to be a tournament.

Let $\mathcal{F}$ be a family of oriented graphs.
We define $\chi_o(\mathcal{F})$ to be the least integer $k$ such that $\chi_o(F) \leq k$ for all $F \in \mathcal{F}$, 
or infinity if no such $k$ exists.
An oriented graph $H$ is a 
\emph{universal target} for $\mathcal{F}$ if
$F \to H$ for every $F \in \mathcal{F}$. 
If $H$ is a universal target for $\mathcal{F}$, then $\chi_o(\mathcal{F}) \leq |V(H)|$.

We now define a useful family of tournaments which serve as universal targets for some families of oriented graphs.
Let $q \equiv 3 \pmod 4$ be a prime power.
The \emph{Paley tournament}, or \emph{non-zero quadratic residue tournament},  $QR_q$,
is defined to have vertex set $\{0,1,2, \dots, q-1\}$ and an arc from $i$ to $j$ if and only if 
$j-i \not\equiv 0 \pmod q$ is a (non-zero) quadratic residue in the field of order $q$.

The Payley tournament $QR_7$ is a universal target for the family of orientations of partial 2-trees, and hence also for the 
family of orientations of outerplanar graphs \cite{SO97}.
The proof of this result 
relies on the fact that $QR_7$ has the property that
for any two different vertices $x$ and $y$, and
any $a, b \in \{+, -\}$, there exists at least one vertex $z$ such that $x \in N^a(z)$ and $y \in N^b(z)$.
Hence, if a vertex $v$ of the oriented graph $G$ has degree at most $2$, then a homomorphism $\phi:(G-v) \to QR_7$
can be extended to a homomorphism of $G$ to $QR_7$, regardless of the orientations of the arcs incident with $v$,  provided that either the neighbours of $v$ have different images under $\phi$, or $v$ is a source or a sink. 

The property in the previous paragraph is generalized as follows.
Let $i$ and $j$ be positive integers. 
A tournament $T$ has \emph{Property $P_{i,1}$} if,  for every $i$-subset $\{x_1, x_2, \ldots, x_i\} \subset V(T)$ 
and every vector $(n_1,n_2, \dots, n_i)$, where each entry $n_r \in \{+, -\}$, 
there exists a vertex $y$ such that $x_r \in N^{n_r}(y)$ for $r = 1, 2, \ldots, i$.
As mentioned above, the Payley tournament $QR_7$ has Property $P_{2, 1}$.
Let $\mathcal{A}$ denote the statement that $x_r \in N^{n_r}(y)$ for $r = 1, 2, \ldots i$.
A tournament $T$ has \emph{Property $P_{i, j}$} if for every $i$-subset $\{x_1, x_2, \ldots, x_i\} \subset V(T)$ 
and every vector $(n_1,n_2, \dots, n_i)$, where each entry $n_r \in \{+, -\}$, 
there are $j$ different vertices $y_1, y_2, \ldots, y_j$ for which statement $\mathcal{A}$ is true.
Property $P_{i,j}$ is related to the subject of \emph{$n$-existentially closed} tournaments (see \cite{BEH81}, \cite{B09} and \cite{BC06}).
The Paley tournaments are a well-studied family of tournaments that have Property $P_{i, j}$ once the number of vertices is sufficiently large \cite{B09}.  

We will eventually show that $QR_7$ is a universal target for the family of 
oriented properly subcubic connected graphs with no source vertex of degree 3 and no sink vertex of degree 3, and 
that $QR_{67}$ is a universal target for the family of properly subquartic connected oriented graphs.
These results will then be used to obtain our bounds on the oriented chromatic number of
subcubic and subquartic graphs.

\section{Preliminary results: oriented cliques}

An \emph{oriented clique}, or  \emph{oclique}, is 
an oriented graph $G$ such that $\chi_o(G) = |V(G)|$. 
The definition arises from the idea that complete graphs are the only graphs for which the chromatic number equals the number of vertices.

Oriented cliques are completely classified by weak diameter.

\begin{thm} [\cite{S14}]
	An oriented graph is an oriented clique if and only if it has weak diameter at most $2$.
\end{thm}

It is therefore easy to check whether a given oriented graph is an oclique.  On the other hand, it is NP-complete to decide whether a given graph can be oriented to be an oclique
\cite{BDS17}.

We now show that the maximum size of an oriented clique with maximum degree $\Delta$ is $O(\Delta^2)$.
Recall that, by contrast, there exist oriented graphs with degree $\Delta$ and oriented chromatic number at least $2^{\Delta / 2}$.  

\begin{prop}
	If $G$ is an oriented clique with $n$ vertices and maximum degree $\Delta$, then $n \leq \frac{1}{2} + \frac{(\Delta+1)^2}{2}$.
\end{prop}

\begin{proof}
	Let $G$ be an oriented clique with maximum degree $\Delta$. We proceed by counting $2$-dipaths in $G$. Since $G$ has weak diameter $2$ each vertex is joined by a 2-dipath (in some direction) to all vertices of $G$ to which it is not adjacent. Therefore the number of $2$-dipaths in $G$ is at least $$\frac{1}{2}\sum_{x\in V(G)} \left({n - d(x) -1}\right),$$ where $d(x)$ is the degree of $x$ in $U(G)$. 
	We further observe that each vertex of $G$ is the midpoint of at most $\lceil \frac{d(x)}{2} \rceil \lfloor \frac{d(x)}{2} \rfloor$  $2$-dipaths. Therefore the number of $2$-dipaths in $G$ is at most $$\sum_{x \in V(G)} \left\lceil \frac{d(x)}{2} \right\rceil \left\lfloor \frac{d(x)}{2} \right\rfloor.$$ Combining together these two observations yields the following inequality.
	
	$$\frac{1}{2}n^2  \leq  \sum_{x\in V(G)}   \frac{d(x)^2}{4} + \frac{d(x) +1}{2} = \frac{n}{4} + \sum_{x\in V(G)} \frac{(d(x)+1)^2}{4}$$
	
	The result follows on observing that $d(x) \leq \Delta$ for all $x \in V(G)$.
\end{proof}

For $\Delta=3$ the above bound yields $8$. However by computer search, neither of the two $3$-regular graphs on $8$ vertices with diameter $2$ can be oriented to be an oclique.
 Figure~\ref{7clique} gives an oriented clique that is an orientation of a properly subcubic graph on $7$ vertices. 
 
 \begin{prop}
 The largest number of vertices in a subcubic oclique is 7.
\end{prop}

For $\Delta = 4$ the above bound yields $13$. 
Computer search by repeated random sampling  failed to yield an oclique on $12$ or $13$ vertices. 
Figure~\ref{7clique} gives an oriented clique that is an orientation of a $4$-regular graph on $11$ vertices. 
We conjecture  the maximum order of an oriented clique with maximum degree $4$ to be $11$.

\begin{figure}
	\begin{center}
		\includegraphics[scale=1]{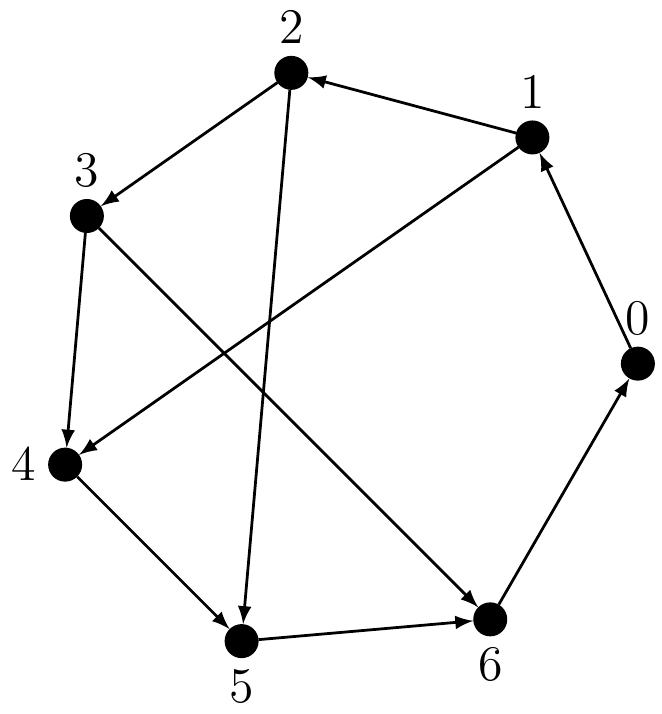}\quad
		\includegraphics[scale=1]{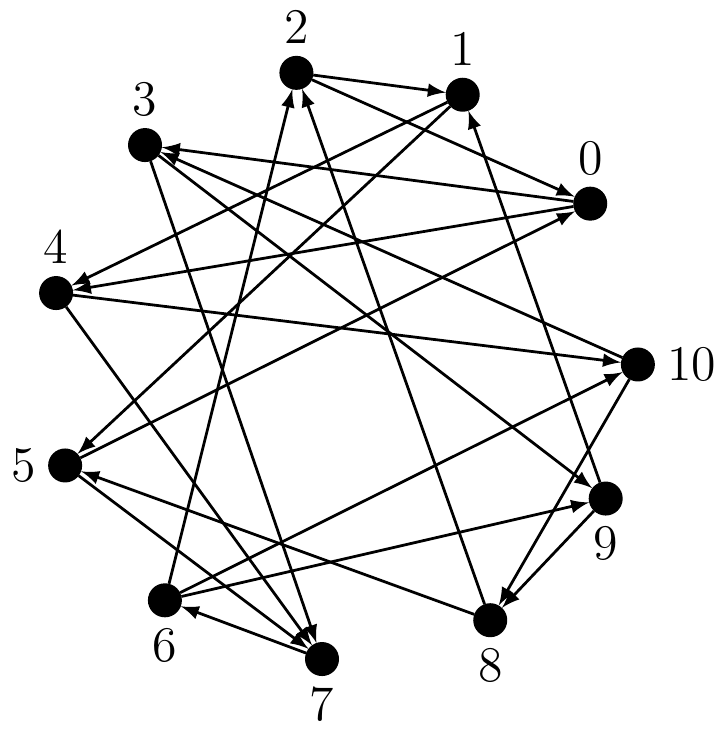}		
	\end{center}
	\caption{Oriented cliques on $7$ and $11$ vertices}
	\label{7clique}
\end{figure}

 \section{Oriented Colourings of Graphs with Maximum Degree Three} \label{cubic}

In this section we  show that every properly subcubic graph that contains neither a source vertex of degree 3 nor a sink vertex of degree 3  admits a homomorphism to $QR_7$. 
We then use this fact to show that every member of $\mathcal{F}_3$ admits an oriented $9$-colouring,
where  $\mathcal{F}_3$  is the family of oriented connected subcubic graphs. 

We begin by observing properties of $QR_7$.  Proofs of properties that are well-known or easy to check are omitted.

\begin{propty} \label{propty:transitive} 
	$QR_7$ is arc transitive and vertex transitive.
\end{propty}

\begin{propty} \label{propty:converse}
	$QR_7$ is isomorphic to its converse.
\end{propty}

\begin{propty}  \label{propty:pij} \cite{GS71}
	$QR_7$ has Property $P_{2,1}$. 
\end{propty}

\begin{propty}\label{propty:backpij}
	For every $x \in V(QR_7)$ 
		\begin{enumerate}
			\item there is a pair of distinct arcs $i_1j_1, i_2j_2 \in E(QR_7)$ such that $i_1,i_2,j_1,j_2$ are all out-neighbours of $x$;
			\item there is a pair of  distinct arcs $i_1j_1, i_2j_2 \in E(QR_7)$ such that $i_1,i_2,j_1,j_2$ are all in-neighbours of $x$; and
			\item there is a pair of distinct arcs $i_1j_1, i_2j_2 \in E(QR_7)$ such that $i_1,j_1$ are out-neighbours of $x$ and $i_2,j_2$ are in-neighbours of $x$.
		\end{enumerate}
\end{propty}

\begin{propty} \label{propty:neighbourhood} 
	For every  $xy\in E(QR_7)$,
	\begin{enumerate}
		\item $x$ and $y$ have exactly one common out-neighbour;
		\item $x$ and $y$ have exactly one common in-neighbour; 
		\item there is exactly one vertex that is an out-neighbour of $x$ and an in-neighbour of $y$; and
		\item there are exactly two vertices that are an in-neighbour of $x$ and an out-neighbour of $y$.
	\end{enumerate}
\end{propty}

The following is an immediate consequence of Property~\ref{propty:transitive}.

\begin{propty} \label{propty:cutedge}
	  An oriented graph $G$ with a cut arc $uv$ admits a homomorphism to $QR_7$ if and only if each component of $G-\{uv\}$ admits a homomorphism to $QR_7$.
\end{propty}

The next proposition follows directly from Property~\ref{propty:pij}.
\begin{propty} \label{propty:subdivide}
	If $G \to QR_7$ and $G^\star$ arises from replacing any arc of $G$ by an oriented path of length two, then $G^\star \to QR_7$.
\end{propty}

Figure~\ref{T4} gives a homomorphism of the three non-isomorphic non-transitive tournaments on four vertices to $QR_7$.  Thus we have:

\begin{propty} \label{propty:nonTrans}
	Every non-transitive tournament on four vertices admits a homomorphism to $QR_7$
\end{propty}

\begin{figure}
	\begin{center}
		\includegraphics[scale=1]{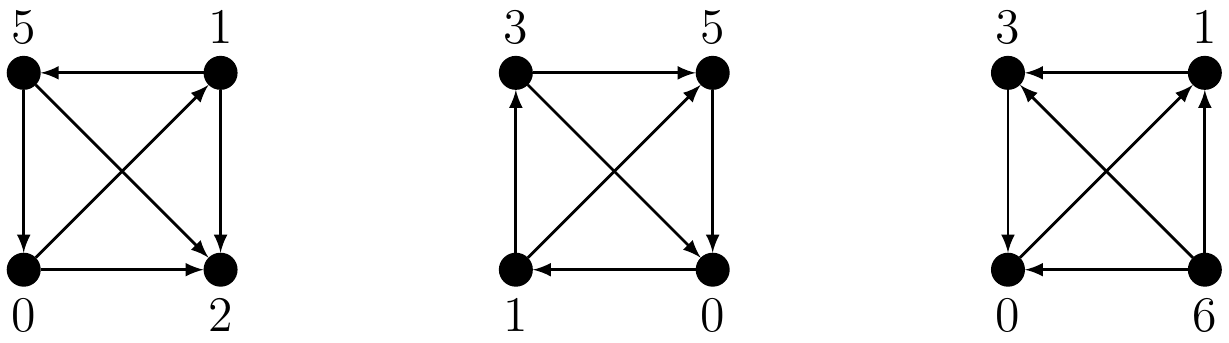}
	\end{center}
	\caption{$QR_7$-colourings of tournaments on four vertices.}
	\label{T4}
\end{figure}

In \cite{SO97} Sopena conjectures that $7$ colours suffice for an oriented colouring of any connected member of $\mathcal{F}_3$. 
While this may be true, it is not the case that $QR_7$ is a universal target for $\mathcal{F}_3$.  For example, the transitive tournament on $4$ vertices does not admit a homomorphism to $QR_7$.

Let $Z_1$, $Z_2$ and $Z_3$ be the set of labelled oriented graphs given in Figure~\ref{orient:family}.
Note that for each oriented graph in the figure, changing the direction of the $2$-dipath between $z_3$ and $z_4$ yields the same oriented graph with a different labelling. Let $Z_i^\prime$ ($1 \leq i \leq 3$) denote these labelled oriented graphs.
The same is true with respect to the $2$-dipath between $z_1$ and $z_2$ in $Z_3$. Let $Z_3^\star$ denote this labelled oriented graph.
For $Z_1$, changing the direction of the arc between $z_1$ and $z_2$ yields the same oriented graph. Let $Z_1^\star$ denote this oriented graph.
Further observe that $Z_2$ and $Z_3$ are each self converse, that is, they are each isomorphic to the oriented graph obtained by changing the direction of every arc.
Let $\mathcal{Z} = \{Z_1, Z_2, Z_3, \overline{Z_1}, \overline{Z_2}, \overline{Z_3}, 
Z_1^\prime, Z_2^\prime, Z_3^\prime, \overline{Z_1^\prime}, \overline{Z_2^\prime}, \overline{Z_3^\prime},
Z_1^\star, Z_3^\star, 
\overline{Z_1^\star}, \overline{Z_3^\star},
  Z_1^{\star\prime}, Z_3^{\star\prime},
 \overline{Z_1^{\star\prime}}, \overline{Z_3^{\star\prime}}\}$.
Note however that up to isomorphism, $\mathcal{Z}$ contains only $Z_1,Z_2, Z_3,$ and $\overline{Z_1}$.
And so when we consider  $Z \in \mathcal{Z}$ we may assume that $Z$ is a labelled copy of $Z_1,Z_2,Z_3$ or $\overline{Z_1}$.

\begin{figure}
	\begin{center}
		\includegraphics[scale=1]{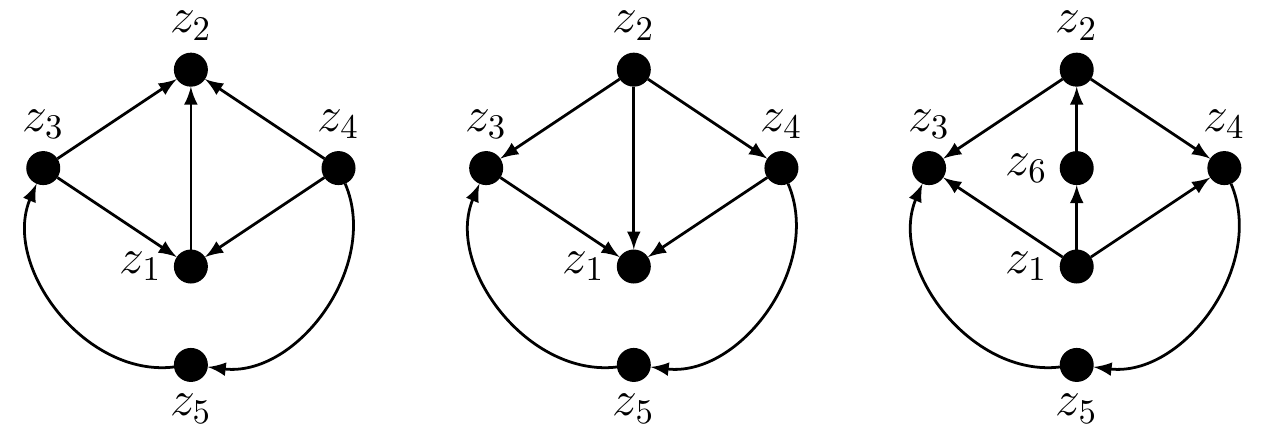}
	\end{center}
	\caption{Oriented graphs $Z_1$, $Z_2$ and $Z_3$ that do not admit a homomorphism to $QR_7$.}
	\label{orient:family}
\end{figure}

\begin{prop} \label{prop:nozed}
	No oriented graph in $\mathcal{Z}$ admits a homomorphism to $QR_7$.
\end{prop}

\begin{proof}
	Let $G$ be an oriented graph in $\mathcal{Z}$ such that there exists $\phi: G \to QR_7$. For each $Z \in \mathcal{Z}$ it must be that $\phi(z_1) \neq \phi(z_2)$. By Property~\ref{propty:neighbourhood} of $QR_7$, $\phi(z_3) = \phi(z_4)$. But $z_3$ and $z_4$ are the ends of a $2$-dipath, a contradiction.
\end{proof}

\begin{cor}
	Any oriented properly subcubic graph that contains a subgraph  from $\mathcal{Z}$ does not admit a homomorphism to $QR_7$.
	\label{nohomfromZ}
\end{cor}

\begin{prop}

If $Z^\prime \notin \mathcal{Z}$ is obtained by replacing an arc of $Z \in \mathcal{Z}$, by a 2-dipath, then $Z^\prime \to QR_7$. 
\label{subdivZ}
\end{prop}

\begin{proof}
Observe that replacing $z_1z_2$ in $Z_1$ by a $2$-dipath yields $\overline{Z_3}$.
Since $Z^\prime \notin \mathcal{Z}$, we may assume that $Z^\prime$ did not arise by this replacement.
	
First suppose the arc replaced by a 2-dipath is incident with $z_5$.
There is a homomorphism $\phi: Z \setminus\{z_5\} \to QR_7$ in which $\phi(z_3) = \phi(z_4)$.
Since every vertex of $QR_7$ is in a directed 3-cycle, $\phi$ can be extended to $Z^\prime$.

Similarly, if the arc replaced by a 2-dipath is incident with $z_6$, then there is a homomorphism
$Z^\prime \to QR_7$. 

Now suppose the arc between $z_1$ and $z_3$ is replaced by a 2-dipath.
Let $Z^{\prime\prime}$ be obtained from $Z$ by reversing the orientation of this arc, and deleting $z_5$ and $z_6$, if it exists.
Then $Z^{\prime\prime}$ is a subgraph of a non-transitive tournament on four vertices.
By Property~\ref{propty:nonTrans}, there is a homomorphism $\phi: Z^{\prime\prime} \to QR_7$
such that $\phi(z_1) \neq \phi(z_3)$ and $\phi(z_3) \neq \phi(z_4)$.
By Property~\ref{propty:pij}, the mapping $\phi$ can be extended to $Z^\prime$.

The remaining cases are similar.
\end{proof}

Consider the family $\mathcal{R}$ of oriented graphs constructed from those in $\mathcal{Z}$ by adding vertices $r_1$ and $r_2$, arcs $r_1z_3$ and $z_4r_2$, and deleting $z_5$.
Observe that any $R \in \mathcal{R}$,  identifying $r_1$ and $r_2$ (into a single vertex) gives the oriented graph from $\mathcal{Z}$ from which $R$ was constructed.

Let $H$ be an oriented properly subcubic graph that contains an element of $\mathcal{R}$ as a subgraph.
If $\phi: H \to QR_7$ is a homomorphism, then we must have $\phi(r_1) \neq \phi(r_2)$;  
otherwise, by the observation in the previous paragraph, $\phi$ implies the existence of a homomorphism from an element of $\mathcal{Z}$ to $QR_7$, contrary to Proposition~\ref{prop:nozed}.

We now describe a reduction that, if it can be applied, transforms a given oriented properly subcubic graph $G$ into an oriented properly subcubic graph $G^R$ with fewer vertices than $G$, and does so such that 
$G \to QR_7$ if and only if $G^R \to QR_7$.  Subsequently, it will therefore suffice to prove the desired result for oriented properly subcubic graphs that can not be reduced.  

\begin{reduction}
	Let $G$ be an oriented properly subcubic graph containing some $R \in \mathcal{R}$ as a subgraph. The properly subcubic graph $G^R$ is obtained from $G$ by
	\begin{itemize}
		\item deleting the vertices corresponding to $z_1,z_2,z_3,z_4$ and, if it exists, $z_6$;
		\item adding a new vertex $r$ together with the arcs $rr_2$ and $r_1r$.
	\end{itemize}
\end{reduction}
We call an oriented properly subcubic graph \emph{reducible} if the above reduction can be applied, i.e., if it contains a subgraph that belongs to $\mathcal{R}$, and  \emph{reduced} otherwise.  
Since each oriented graph in $\mathcal{R}$ contains either a source of degree 3 or a sink of degree $3$,  an oriented properly subcubic graph with no source of degree 3 and no sink of degree $3$ is reduced.
Observe that if $G$ is connected and $G^R$ is disconnected, then the copy of $R$ contains the vertex $z_6$ and $G$ contains a cut arc incident with $z_6$ (as all other vertices of $R$ except $r_1$ and $r_2$ have degree 3).

\begin{lem} [The Reduction Lemma]
	Let $G$ be a reducible oriented properly subcubic graph. The oriented graph $G$ admits a homomorphism to $QR_7$ if and only if $G^R$ admits a homomorphism to $QR_7$.
\end{lem}

\begin{proof}
It suffices to consider oriented connected graphs.

	Let $G$ be a reducible oriented properly subcubic graph that admits a homomorphism $\phi$ to $QR_7$. 
	Let $x_i$ be the vertex corresponding to $z_i$ in the copy of $Z \in \mathcal{Z}$ constructed by identifying $r_1$ and $r_2$ in $G$. 
	Since $x_1$ and $x_2$ are adjacent,  $\phi(x_1) \neq \phi(x_2)$. 
	By Property~\ref{propty:neighbourhood}  we have $\phi(x_3) = \phi(x_4)$, which in turn implies that $\phi(r_1) \neq \phi(r_2)$. 
	Restricting $\phi$ to  $V(G^R) \cap V(G)$, and then extending it to $r$ by using Property~\ref{propty:pij}, yields a homomorphism of $G^R$ to $QR_7$.
	
	Assume there homomorphism $\beta: G^R \to QR_7$. By Property~\ref{propty:transitive}, we may assume that $\beta(r) = 0$. 
	
Suppose the vertex $x_6$ does not exist in $G$.
Then   the restriction of $\beta$ to $V(G) \cap V(G^R)$ can be extended to a homomorphism $G \to QR_7$.  
Note that $r, x_1, x_2, x_3, x_4 \not\in V(G)$.
Map each of $x_3$ and $x_4$ to $0$ and the remaining vertices using Property~\ref{propty:pij}. 	

Now suppose $x_6 \in V(G)$. 

If $x_6$ has degree 2 in $G$, then the restriction of $\beta$ to $V(G) \cap V(G^R)$ can be extended to a homomorphism $G \to QR_7$ by mapping $x_6$ to 0 and proceeding as above. 

So, suppose $x_6$ has degree 3 in $G$, then let $s \not\in \{x_1, x_2\}$ be the third vertex to which it is adjacent in $U(G)$. 
Since $s \in G^R$, $\beta(s)$ is defined.  Let $\beta(s) = k$.
Our goal is to extend  the restriction of $\beta$ to $V(G) \cap V(G^R)$ to all vertices of $G$ so that the direction of the arc between $x_6$ and $s$ is the same as that between $\beta(x_6)$ and $\beta(s)$.
By Property~\ref{propty:converse} of $QR_7$ we may, without loss of generality, assume  the arcs joining $x_1$ and $x_3$, and $x_1$ and $x_4$, are oriented such that $x_3x_1, x_4x_1 \in E(G)$. 
For every vertex of $QR_7$ there exists an element of $\{3,5,6\}$ for which it an out-neighbour, and an element of $\{2,3,5\}$ for which it is an in-neighbour.
Thus it is possible to choose an image for $x_6$ so that the orientation of the arc joining it and $s$ is preserved.
For any choice of image of $x_6$, there is a 2-dipath that starts in $\{1, 2, 4\}$ and ends in $\{1, 2, 4\}$ for which it is the midpoint.
Map $x_2$ and $x_1$ to the start and end of such a 2-dipath, respectively.
Finally, map $x_3$ and $x_4$ to 0.
This completes the proof.
\end{proof}

Consider any oriented graph $G$ with the property that a single reduction necessarily produces an oriented graph from $\mathcal{Z}$. 
Without loss of generality, we may assume that, in $G^R$, the vertex $r$ corresponds to $z_5$ in such a graph.
Therefore $G$ contains the configuration shown in Figure~\ref{orient:singlereduce}, or any one constructed by replacing one or both of the $2$-dipaths $x_4x_5x_3$ and $y_4y_5y_3$ with a single arc from its start vertex to its end vertex.  
In the figure, 
the 
undirected edges are assumed to be oriented so the resulting oriented graph is reducible to only a graph from $\mathcal{Z}$.

\begin{figure}
	\begin{center}
		\includegraphics{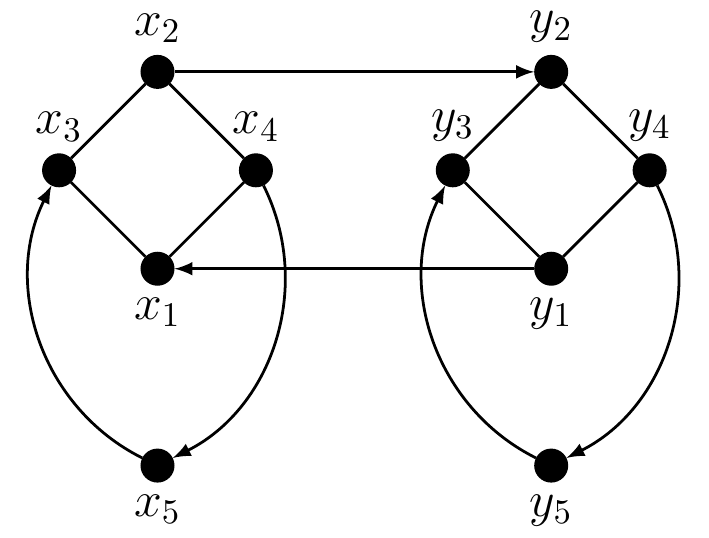}
	\end{center}
	\caption{A graph that reduces to a graph containing a member of $\mathcal{Z}$ with a single reduction.}
	\label{orient:singlereduce}
\end{figure}

We now show that any reduced oriented properly subcubic graph that does not have a subgraph  from $\mathcal{Z}$ admits a homomorphism to $QR_7$. In the sequel we use this result to prove that every oriented connected cubic graph admits a homomorphism to a tournament on at most $9$ vertices.
We begin with a technical lemma.

\begin{lem}
Let $G$ be an oriented graph with a vertex $w$ of degree 2 such that $G$ contains a subgraph $R \in \mathcal{R}$ 
that contains $w$.
If there is a sequence of reductions of $G$ that produces a reduced graph that contains  a subgraph belonging to $\mathcal{Z}$, then $G - w$ is reducible.
\label{reducetoZ} 
\end{lem}

\begin{proof}
Suppose $G - w$ is reduced.  We shall obtain a contradiction.

Since $G$ can be reduced to contain a subgraph belonging to $\mathcal{Z}$, in the penultimate step
the graph $G^\prime$ arising from the sequence of reductions contains the configuration in Figure
\ref{orient:singlereduce}.  
Notice that the only two vertices in the figure that could have been introduced in the sequence of reductions 
are those of degree 2.

Reduce $G$ using the copy of $R$ in the hypothesis. 
 If $G^R$ contains a subgraph belonging to $\mathcal{Z}$, then the configuration in Figure
\ref{orient:singlereduce} is a subgraph of $G$.  
At most one of the degree 2 vertices in the figure is $w$.
If it is deleted, the figure still contains a subgraph belonging to $\mathcal{R}$.
Therefore $G - w$ is reducible, a contradiction.

Otherwise, $G^R$ is reducible, and hence contains a subgraph belonging to $\mathcal{R}$.
Let $r^\prime$ denote the vertex introduced in the reduction that produced $G^R$.
Note that $r^\prime$ has degree 2.
If there is a subgraph belonging to $\mathcal{R}$ that does not contain the vertex $r^\prime$, 
then it is a subgraph of $G-w$, a contradiction.
Therefore it contains $r^\prime$.
Since $r^\prime$ has degree 2, it can correspond to $r_1, r_2$ or $z_6$ in this subgraph.

If it corresponds to $r_1$ or $r_2$, then after the reduction it still has degree 2 and is joined by an arc
(with some orientation) to the vertex of degree 2 introduced in the reduction.  
Notice that 
the configuration in Figure~\ref{orient:singlereduce} has all vertices of degree at least 2, and no 
adjacent vertices of degree 2.  
Therefore, no vertex of degree 2 that corresponds to $r_1$ or $r_2$ in a subgraph of $\mathcal{R}$
used in a reduction can be involved in the configuration in Figure~\ref{orient:singlereduce}.

Suppose $r^\prime$ corresponds to $z_6$.  
When the reduction is applied, this vertex is deleted and a new vertex of degree 2 is introduced.

Before the final reduction that produces a subgraph in $\mathcal{Z}$, the graph $H$ resulting from
the sequence of reductions contains the configuration in Figure~\ref{orient:singlereduce}.
By the above argument, at most one vertex of this configuration is not a vertex of $G$,
and that vertex has degree 2.
If that vertex is deleted, the figure still contains a subgraph belonging to $\mathcal{R}$.
Further, all of the vertices of this subgraph are vertices of $G-w$.
Hence $G-w$ is reducible, a contradiction.

\end{proof}

\begin{lem} \label{lem:heavyLift}
	Every reduced oriented connected properly subcubic graph that does not contain a subgraph isomorphic to an oriented graph in $\mathcal{Z}$ admits a homomorphism to $QR_7$.
\end{lem}

\begin{proof}
Let $G$ be a minimum counter-example  with respect to number of vertices and, subject to that, with respect to  number of arcs.  
If there is a vertex $x$ of degree 1, then a homomorphism $(G-x) \to QR_7$ can be extended to $G$.
Hence by hypothesis and choice of $G$ there is a vertex $z$ of degree 2.  Let $u$ and $v$ be the neighbours of $z$ in $U(G)$.
If  $z$ has in-degree 0 or out-degree 0 in $G$, then by Property~\ref{propty:pij} a homomorphism of $G-z$ to $QR_7$ can be extended to $G$, a contradiction.
Hence, without loss of generality, $u, z, v$ is a 2-dipath in $G$.
By Property~\ref{propty:pij} (again) and minimality, 
it must be that in every homomorphism $\phi: (G -z) \to QR_7$ we have $\phi(u) = \phi(v)$, otherwise $\phi$ can be extended to $z$.
Finally, note that if either $u$ or $v$ has degree one in $G - z$, then we need not have $\phi(u) = \phi(v)$. Hence both of these vertices have degree 2 in $G-z$.

Let $u_1$, $u_2$ (respectively $v_1$ and $v_2$) be the neighbours of $u$ (respectively $v$) in $G -z$.   
We proceed by establishing properties that $G$ must have, and then eventually show that no such graph $G$ exists. 
	
\begin{claim} \label{claim:cutclaim}
		Neither $G$ nor $G -z $ contains a cut arc.
\end{claim}
Property~\ref{propty:cutedge} and minimality imply that $G$ has no cut arc.  

Suppose $G -z$ contains a cut arc, $e$. 

Suppose $u$ and $v$ are in the same component of $(G - z) - e$.
By Property~\ref{propty:cutedge}, $(G-z) \to QR_7$ if and only if both components of $(G-z)-e \to QR_7$.
It follows that, as above, $u$ and $v$ must have the same image in any such homomorphism.
But then the mapping can not be extended to $z$.
Hence there is no homomorphism $(G-e) \to QR_7$, contrary to the minimality of $G$.

Now suppose that $u$ and $v$ are in different components of $(G-z)-e$.  
Let $A_u$ be the component of $(G-z)-e$ containing $u$ and let $A_v$ be the component containing $v$.
Let $e = xy$,  where $x$ is a vertex of $A_u$ and $y$ is a vertex of $A_v$.
By the minimality of $G$, there are homomorphisms  $\phi: A_u \to QR_7$ and $\rho: A_v \to QR_7$.
By Properties~\ref{propty:transitive} and~\ref{propty:converse}, we may assume $\phi(x) = 0$ and $\phi(u) \in \{0,1\}$.
Similarly,  we may assume $\rho(y) = 2$ and $\rho(v) \in \{2,3,5\}$. 
But then the function $\psi: (G-z) \to QR_7$ defined by
$$\psi(w) = \begin{cases}
\phi(w) & \mathrm{if}\ w \in V(A_u),\\
\rho(w) &  \mathrm{if}\ w \in V(A_v),\\
\end{cases}$$
is a homomorphism $(G-z) \to QR_7$ such that $\psi(u) \neq \psi(v)$, a contradiction.

\begin{claim} \label{claim:2edgecut}
		If $\{e_1,e_2\}$ is a minimal edge cut in $G-z$, then  $e_1$ and $e_2$ have a common endpoint of degree $2$.
\end{claim}
	
Suppose the contrary.  Then either $e_1$ and $e_2$ have a common endpoint of degree 3, or have no common endpoint.  We consider these cases in turn.
	
\emph{Case I: $e_1$ and $e_2$ have a common endpoint of degree $3$.} \\
Let $a$ be the hypothesized common endpoint of $e_1$ and $e_2$. 
If $b$ is the neighbour of $a$ (in $U(G)$) that is incident with neither  $e_1$ nor $e_2$, 
then it is easy to see that  $ab$ is a cut arc of $G-z$, contrary to Claim~\ref{claim:cutclaim}.
	
\emph{Case II: $e_1$ and $e_2$ do not have a common endpoint.}  \\
Since neither $e_1$ nor $e_2$ is a cut edge of $G-z$, the oriented graph  $(G - z) - \{e_1,e_2\}$ has exactly two components. 
Let $a_1$ and $b_1$ be the endpoints of the edge of $U(G)$ corresponding to $e_1$,
and 
let $a_2$ and $b_2$ be the endpoints of  the edge of $U(G)$ corresponding to $e_2$.
Assume  $a_1$ and $a_2$ are in the same component $(G-z)-\{e_1,e_2\}$, call it $A$. 
Let $B = (G-z) - A$.

Either $u$ and $v$ belong to the same component of $(G-z)-\{e_1,e_2\}$, or they belong to different components of $(G-z)-\{e_1,e_2\}$.  This leads to the following two subcases.	
	
\emph{Subcase II.i: $u$ and $v$ are in different components of $(G-z)-\{e_1,e_2\}$}.\\
Without loss of generality, assume $u \in V(A)$. 
By the minimality of $G$, there exist  homomorphisms  $\phi_A: {A \to QR_7}$ and $\phi_B: B \to QR_7$.
We consider possibilities depending on whether  $\phi_A(a_1) = \phi_A(a_2)$ and $\phi_B(b_1) = \phi_B(b_2)$.
 
 \emph{Subcase II.i.i} $\phi_A(a_1) \neq \phi_A(a_2)$ and $\phi_B(b_1) \neq \phi_B(b_2)$.\\
By  Property~\ref{propty:transitive}, it can be assumed that  $\phi_A(a_1) = 0$ and $\phi_A(a_2)=1$. 
 Since $\phi_B(b_1) \neq \phi_B(b_2)$, we may assume the existence of an arc between $b_1$ and $b_2$ (with some orientation). If such an arc does not exist, then we may add it so that it has the same orientation as the arc between $\phi_B(b_1)$ and $\phi_B(b_2)$ in $QR_7$. 
Table~\ref{claim:2edgecuttable} gives the possibilities for the orientations of an arc between $b_1$ and $b_2$
and partial homomorphisms $\alpha_B$ and $\alpha_B^\prime$ of $B$ to $QR_7$.  By Property~\ref{propty:transitive}, each of these can be extended to a homomorphism $B \to QR_7$.  In what follows we will assume $\alpha_B$ and $\alpha_B^\prime$ are homomorphisms. 
For each possibility, using $\phi_A, \alpha_B, \alpha^\prime_B$, and Property~\ref{propty:transitive}, we construct homomorphisms $\phi: (G - z) \to QR_7$ and $\phi^\prime: (G - z) \to QR_7$ by
$$\phi(w) = \begin{cases}
\phi_A(w)&  \mathrm{if}\ w \in V(A),\\
\alpha_B(w) &  \mathrm{if}\ w \in V(B),\\
\end{cases}
\quad\quad
\phi^\prime(w) = \begin{cases}
\phi_A(w)&  \mathrm{if}\ w \in V(A),\\
\alpha^\prime_B(w) &  \mathrm{if}\ w \in V(B).\\
\end{cases}
$$

By Property~\ref{propty:transitive}, there is an automorphism, $\sigma$, of $QR_7$ such that 
$\alpha^\prime_B(\sigma(b_1)) = \alpha_B(b_1)$ and 
$\alpha^\prime_B(\sigma(b_2)) = \alpha_B(b_2)$.
In particular, we can choose $\sigma: QR_7 \to QR_7$ to be defined by
$$\sigma(w) = \frac{\alpha^\prime_B(b_2)-\alpha^\prime_B(b_1)}{\alpha_B(b_2)-\alpha_B(b_1)}(w  - \alpha_B(b_1))  + \alpha^\prime_B(b_1) \pmod 7,$$
so that, in addition, $\sigma$ fixes no vertex of $QR_7$.	
	\begin{table}
		
		\begin{center}
			\includegraphics[scale=1]{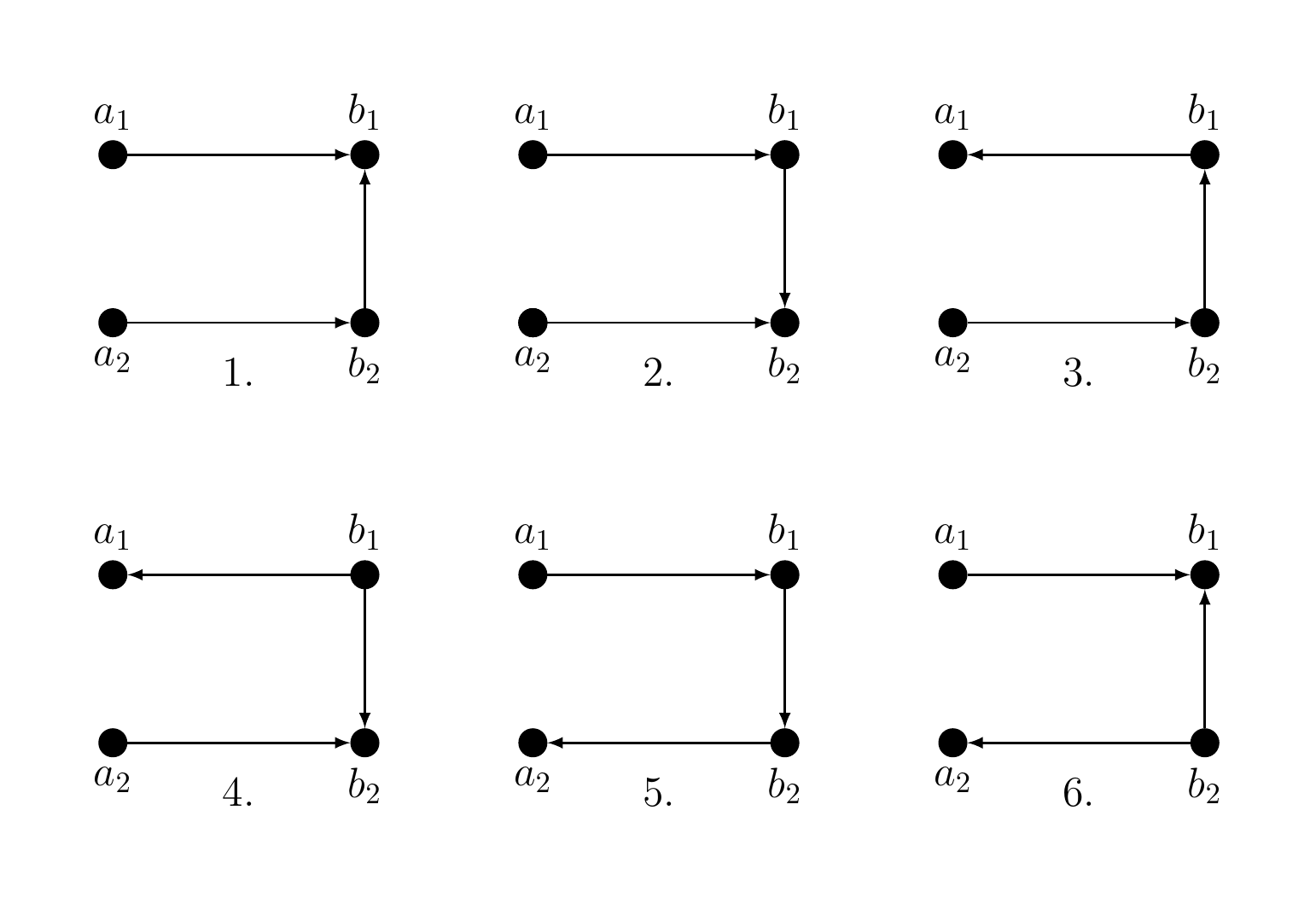}
		\end{center}
		\begin{center}
			\begin{tabular}{ccc}
				
				\begin{tabular}{c|c|c}
					$1.$ & $\alpha_B$ & $\alpha^\prime_B$ \\ 
					\hline $b_1$ & $4$ & $2$ \\ 
					\hline $b_2$ & $2$ & $5$ \\ 
				\end{tabular} 
				\vspace{0.5cm}

				&  
				\begin{tabular}{c|c|c}
					$2.$ & $\alpha_B$ & $\alpha^\prime_B$ \\ 
					\hline $b_1$ & $1$ & $2$ \\ 
					\hline $b_2$ & $2$ & $3$ \\ 
				\end{tabular}

				&
				\begin{tabular}{c|c|c}
					$3.$ & $\alpha_B$ & $\alpha^\prime_B$ \\ 
					\hline $b_1$ & $3$ & $6$ \\ 
					\hline $b_2$ & $2$ & $5$ \\ 
				\end{tabular}

				\\ 
				\begin{tabular}{c|c|c}
					$4.$ & $\alpha_B$ & $\alpha^\prime_B$ \\ 
					\hline $b_1$ & $5$ & $6$ \\ 
					\hline $b_2$ & $2$ & $3$ \\ 
				\end{tabular} 
				\vspace{0.5cm}
				&  
				\begin{tabular}{c|c|c}
					$5.$ & $\alpha_B$ & $\alpha^\prime_B$ \\ 
					\hline $b_1$ & $2$ & $4$ \\ 
					\hline $b_2$ & $4$ & $6$ \\ 
				\end{tabular}

				&
				
				\begin{tabular}{c|c|c}
					$6.$ & $\alpha_B$ & $\alpha^\prime_B$ \\ 
					\hline $b_1$ & $2$ & $1$ \\ 
					\hline $b_2$ & $0$ & $6$ \\ 
				\end{tabular}  \\ 
			\end{tabular} 
		\end{center}		
		\caption{
		}
		\label{claim:2edgecuttable}
	\end{table}

Notice that, with this choice of $\sigma$,  if $\phi(u) = \phi(v)$, then $\phi^\prime(u) \neq \phi ^\prime(v)$, contradicting that $u$ and $v$ have the same image in every homomorphism  $(G-z) \to QR_7$.
	
\emph{Subcase II.i.ii}  $\phi_A(a_1) = \phi_A(a_2)$ and $\phi_B(b_1) = \phi_B(b_2)$. \\
Suppose the arcs between $A$ and $B$ (i.e., $e_1$ and $e_2$) are all oriented from $A$ to $B$.
Observe that the oriented graph produced by identifying $a_1$ and $a_2$, and $b_1$ and $b_2$  has a cut arc, $e$, and
each component of $(G-z) - e$ admits a homomorphism to $QR_7$.   
By Property~\ref{propty:cutedge}, there is a homomorphism $\phi: (G-z) \to QR_7$, and since $u$ and $v$ are in different components of $(G-z) - \{e_1, e_2\}$, the mapping $\phi$ can be defined so that $\phi(u) \neq \phi(v)$, a contradiction.

Hence suppose, without loss of generality, that $e_1 = b_1a_1$ and $e_2 = a_2b_2$.
Construct $A^\star$ from $A$ by adding a vertex, $a$, together with the arcs $a_2a$ and $aa_1$, and similarly construct $B^\star$ from $B$.
If $A^\star \to QR_7$, then $a_1$ and $a_2$ have different images in any homomorphism $A^\star \to QR_7$ (as they are joined by a 2-dipath), and similarly for $B^\star$.
If both $A^\star \to QR_7$ and $B^\star \to QR_7$, then Subcase \emph{II.i.i} applies.
Therefore either $A^\star \not\to QR_7$ or $B^\star \not\to QR_7$.
Without loss of generality, $A^\star \not\to QR_7$.
Then, by minimality, either $A^\star$ contains an element $Z \in \mathcal{Z}$, or $A^\star$ is reducible to an oriented graph that contains an element of $\mathcal{Z}$ as a subgraph.

Suppose $A^\star$ contains an element $Z \in \mathcal{Z}$.
Then $a \in V(Z)$. 
Since $a$ has degree $2$ (in $U(A^\star)$), we can assume that it corresponds to $z_5$  (note that $z_5$ and $z_6$ are interchangeable, if $z_6$ exists).
But then $G-z$ has a subgraph that is in $\mathcal{R}$: it consists of $Z-a$ together with $b_1$ and $b_2$ (corresponding to $r_2$ and $r_1$, respectively).
Therefore $G- z$ is reducible,  a contradiction.

Now suppose $A^\star$ is reducible to an oriented graph that contains an element of $\mathcal{Z}$ as a subgraph.
Let $R \in \mathcal{R}$ be a subgraph of $A^\star$.
Since $G-z$ is not reducible, $a \in V(R)$.
Further, in $R$ the vertex $a$ plays the role of one of $r_1, r_2,$ or $z_6$ (as all other vertices have degree 3 in $U(R)$).
If $a$ plays the role of $r_1$ or $r_2$, then using $b_1$ or $b_2$ instead (respectively) gives a copy of $R$ in $G-z$, a contradiction.  Hence $a$ plays the role of $z_6$.
But now, by Lemma~\ref{reducetoZ}, $A^\star - a = A$ is reducible, a contradiction.

\emph{Subcase II.i.iii}  $\phi_A(a_1) = \phi_A(a_2)$ for all homomorphisms $\phi_A: A \to QR_7$, and  $\phi_B(b_1) \neq \phi_B(b_2)$ for all homomorphisms  $\phi_B$: $B \to QR_7$.\\
 Since $\phi_B(b_1) \neq \phi_B(b_2)$, we may assume the existence of an arc between $b_1$ and $b_2$ (with some orientation). If such an arc does not exist, then we may add it so that it has the same orientation as the arc between $\phi_B(b_1)$ and $\phi_B(b_2)$ in $QR_7$.  By identifying $a_1$ and $a_2$ into a single vertex and applying Property~\ref{propty:backpij} we obtain a homomorphism of $G - z$ to $QR_7$ in which $u$ and $v$ have different images,
 a contradiction.
	
	\emph{Subcase II.ii: $u$ and $v$ are in the same component of $(G-z)-\{e_1,e_2\}$}. \\
Suppose $u, v \in V(A)$.  By the minimality of $G$, observe that  $B$ admits a homomorphism to $QR_7$. Construct $A_z$ by adding the vertex $z$ together with the arcs $uz$ and $zv$ to $A$. By the minimality of $G$, $A_z$ admits a homomorphism to $QR_7$. Regardless of the orientations of the arcs between $A$ and $B$, these homomorphisms may be combined to be one of $G$ to $QR_7$  as above, as long as it is not the case that for all $\phi_{A_z}: A_z \to QR_7$ and $\phi_B$: $B \to QR_7$ we have that  $\phi_A(a_1) = \phi_A(a_2)$ and $\phi_B(b_1) = \phi_B(b_2)$, and that, without loss of generality, $a_1$ is the head of $e_1$ and $a_2$ is the tail of $e_2$. However in this case we proceed as in Subcase \emph{II.i.ii} when $e_1=b_1a_1$ and $e_2=a_2b_2$. We  construct $A^\star$ and $B^\star$ and conclude that either $G-z$ is reducible or contains a subgraph from $\mathcal{Z}$, a contradiction.
	
Therefore if $\{e_1,e_2\}$ is an edge cut in $G-z$, then  $e_1$ and $e_2$ have a common endpoint of degree $2$.
	
	\begin{claim} \label{claim:deg2claim}
		$G$ contains a single vertex of degree $2$.
	\end{claim}
Suppose there exists $z^\prime\neq z$ with neighbours $u^\prime$ and $v^\prime$.  
	
Suppose $z^\prime$ has in-degree 0 or out-degree 0.	
By minimality, $(G-z^\prime) \to QR_7$, and by Property~\ref{propty:pij} any such homomorphism can be extended 
to $G$, a contradiction.
	
Suppose, then, that $u^\prime z^\prime v^\prime$ is a $2$-dipath in $G$. 
Consider the oriented graph $G^\prime$ obtained by deleting $z^\prime$ 
and adding the arc $u^\prime v^\prime$. 
Any homomorphism $G^\prime \to QR_7$ can be extended to 
$G$ by Property~\ref{propty:pij}, a contradiction.
Hence $G^\prime \not\to QR_7$.
Therefore, by minimality of $G$, the oriented graph  $G^{\prime}$ is either a subgraph in $\mathcal{Z}$, or is reducible to an oriented graph that contains a subgraph in $\mathcal{Z}$.

Suppose $G^{\prime}$ contains a copy of $Z \in \mathcal{Z}$.
Certainly $u^\prime v^\prime$ appears in this copy of $Z$.
If this arc does not correspond to the one between $z_1$ and $z_2$ in $Z_1$ or $\overline{Z_1}$, then reversing the orientation of this arc, i.e., adding arc $v^\prime u^\prime$ and removing $v^\prime u^\prime$  does not yield a copy of any element of $\mathcal{Z}$.
Nor can it yield a reducible graph.
Thus, we have $Z = Z_1$ or $Z=\overline{Z_1}$ and the arc between $u^\prime$ and $v^\prime$ corresponds to the one between $z_1$ and $z_2$.
However, in this case we note that $G$ has a subgraph isomorphic to $Z_3$, a contradiction.

Now suppose $G^\prime$ contains a subgraph $R \in \mathcal{R}$ and is reducible to an oriented graph that contains a subgraph in $\mathcal{Z}$.  By definition of $G$, we have $u^\prime v^\prime \in E(R)$.  If $u^\prime$ or $v^\prime$ corresponds to $r_1$ or $r_2$, then $G$ contains a copy of $R$ with $z^\prime$ playing that role, a contradiction.  
Otherwise,  $(G^\prime)^R$ is either reduced, or is reducible to an oriented graph that contains a subgraph in $\mathcal{Z}$.  We consider these cases in turn.

In the former case, $G^\prime$ contains the configuration in Figure~\ref{orient:singlereduce}.    
Since $G$ does not contain this configuration and is reduced, it must be that the newly added arc corresponds to the arc between $x_2$ and $y_2$ or the arc between $y_1$ and $x_1$. 
However,  if $u^\prime v^\prime$ is replaced with the  2-dipath $u^\prime, z^\prime, v^\prime$, then it is easy to see that $G$ is reducible,  a contradiction.

In the latter case, if $r$ is the vertex of $(G^\prime)^R$ introduced in the reduction, then by Lemma~\ref{reducetoZ}, 
$(G^\prime)^R - r$ is reducible.  Since $(G^\prime)^R - r$ is a subgraph of $G$, it follows that $G$ is reducible,
a contradiction.

	\begin{claim} \label{claim:indep}
		$u_1$ and $u_2$ are not adjacent in $G - z$.
	\end{claim}
	If $u_1$ and $u_2$ are adjacent in $G- z$, then the arcs respectively incident with $u_1$ and $u_2$ that are not incident with $u$ form a $2$-edge cut. By Claim~\ref{claim:2edgecut}, these arcs have a common endpoint of degree $2$. Since $u$ and $v$ are the only vertices of degree $2$ in $G-z$ this common endpoint must be $v$. Therefore $G$ contains only the vertices $u,v,z,u_1,u_2$. 
	Since $G$ contains no copy of a graph from $\mathcal{Z}$, it must be that $G$  is a subdivision of a non-transitive tournament on four vertices.  
	By Property~\ref{propty:nonTrans} of $QR_7$, $G$ admits a homomorphism to $QR_7$, a contradiction.

	\begin{claim} \label{claim:sourceSink}
		Each of $u$ and $v$ is either a source or sink vertex in $G -z$.
	\end{claim}
	Suppose the contrary. 
	That is, suppose  $u_1uu_2$ forms a $2$-dipath in $G -z$. Consider the graph, $H_u$, constructed from $G -z$ by removing $u$ and adding the arc $u_2u_1$. 
	If $H_u$ admits a homomorphism to $QR_7$,  then $G - z- u$ admits a homomorphism to $QR_7$, as $G - z- u$ is a subgraph of $H_u$.
	By part $(4)$ of Property~\ref{propty:neighbourhood} of $QR_7$,  such a homomorphism can be extended in two different ways to include $u$.
	In particular, it can be extended so that $u$ and $v$ have different images, a contradiction.
	Thus it must be that $H_u$ does not admit a homomorphism to $QR_7$.
	Therefore we may assume that $H_u$ contains either a subgraph from $\mathcal{Z}$ or it  contains a subgraph from $\mathcal{R}$ and is reducible to an oriented graph that contains a subgraph from $\mathcal{Z}$.
	We will consider these cases in turn. In each case we will derive a contradiction by constructing a homomorphism in which $u$ and $v$ do not have the same image.
	
	\emph{Case I: $H_u$ contains a subgraph  $Z \in \mathcal{Z}$.}\\
	To construct a homomorphism of $G-z$ to $QR_7$ in which $u$ and $v$ have different images we first show that $Z= Z_3$ and that the vertices corresponding to $z_5$ and $z_6$ in $Z$ each have degree $3$ in $G-z$. 
	
	We begin by showing that $Z \neq Z_1,Z_2, \overline{Z_1}$.
	Observe that $v$ is not contained in $H_u$.
	If $z_6$ does not exist, then the vertex corresponding to $z_5$ in $H_u$ in much have a third neighbour in $G-z$ that is not contained in $Z$.
	Otherwise there is no path from $u$ to $v$ in $G-z$, as the only vertices reachable from $u$ would be those appearing in the copy of $Z$. 
	However in this case we observe that an arc incident in $z_5$ in $G-z$ is a cut arc.
	This contradicts Claim \ref{claim:cutclaim}.
 	Therefore $z_6$ exists, and as such $Z = Z_3$.

	We now show that vertices corresponding to $z_5$ and $z_6$  each have degree $3$ in $G-z$.
	If the vertices corresponding respectively to $z_5$ and $z_6$ have degree $2$ in $G-z$, then $z_5$ and $z_6$  must correspond, in some order, to $u$ and $v$, as these are the only vertices of degree $2$ in $G-z$. 
	However both $z_5$ and $z_6$  are contained in $H_u$. 
	This is contradiction, as $v$ is not contained in $H_u$.
	
	If the vertex corresponding to $z_5$ has degree $2$ in $G-z$ and the vertex corresponding to $z_6$ has degree $3$, then the arc incident with the vertex corresponding to $z_6$ that is not contained in $Z$ is a cut arc in $G-z$, contradicting Claim~\ref{claim:cutclaim}.
	A similar argument applies when the vertex corresponding to $z_5$ has degree $3$ and the vertex corresponding to $z_6$ has degree $2$ in $G-z$.
	Therefore vertices corresponding respectively to $z_5$ and $z_6$ each have degree $3$ in $G-z$.

	We now derive a contradiction by constructing a homomorphism of $G-z$ to $QR_7$ where $u$ and $v$ have different images. 
	Consider the arcs incident with the vertices corresponding to $z_5$ and $z_6$.
	Each of $z_5$ and $z_6$ are incident with an arc not contained in $Z$.
	These arcs form a $2$-edge cut in $G-z$.
	By Claim~\ref{claim:2edgecut}, these arcs have a common endpoint of degree $2$.
	This common endpoint must be $v$, as, by Claim~\ref{claim:deg2claim}, $G-z$ contains no other vertices of degree $2$.
	Since $H_u$ contains $Z$, but $G-z$ does not contain $Z$, the arc $u_2u_1$ corresponds  to an arc in some graph from $\mathcal{Z}$ that contains $z_6$.
	This implies $G-z$ is one of the partially oriented graphs in Figure~\ref{orient:not2dp}, or one obtained by reversing each of the oriented edges.

	In each of these oriented graphs we see a partial $QR_7$-colouring.
	In all of these graphs other than the first, this colouring can be extended so that $u$ and $v$ are assigned different colours regardless of the orientations of the arcs incident with $v$.
	To see this, note that $u$ can receive colour $1$.
	Regardless of the orientation of the arcs incident with $v$, $v$ cannot be coloured with $1$, as one of its neighbours has this colour.
	Thus, using Property~\ref{propty:pij} of $QR_7$, these colourings can be extended so that $u$ and $v$ receive different colours.
	
			\begin{figure}
		\begin{center}
			\begin{tabular}{cccc}
				\includegraphics[width = 0.25\linewidth]{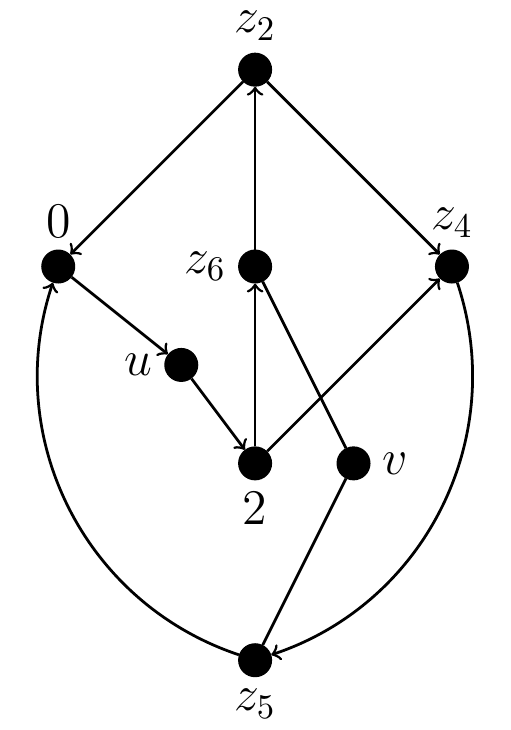} &
				\includegraphics[width =0.25\linewidth]{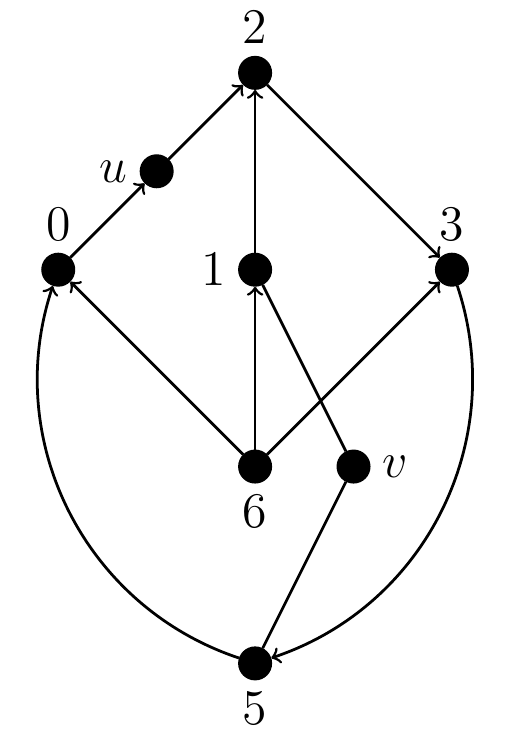} &
				\includegraphics[width = 0.25\linewidth]{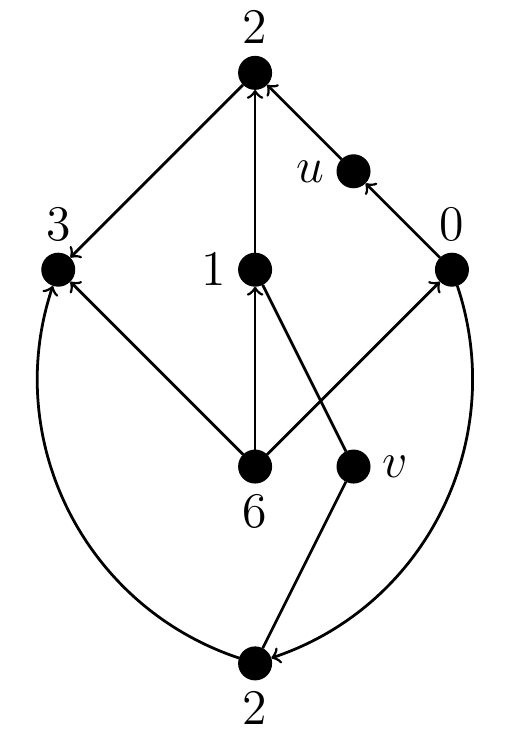}  &
				\includegraphics[width =0.25\linewidth]{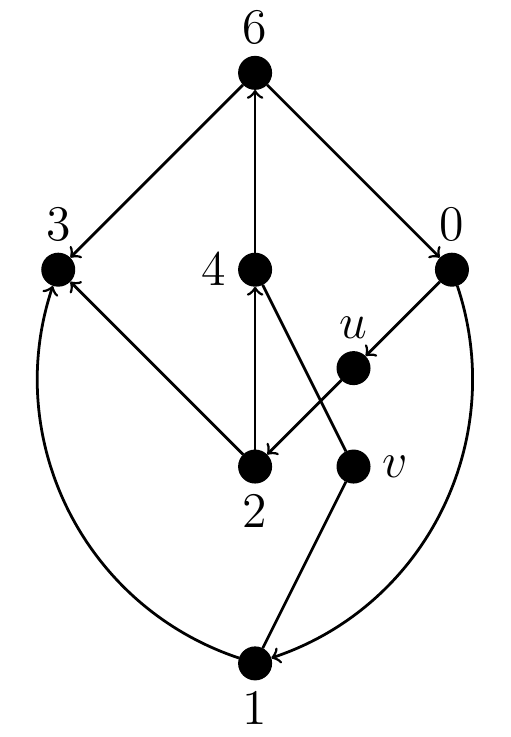} 
				
			\end{tabular}
			\begin{tabular}{cccc}
				\includegraphics[width  =0.25\linewidth]{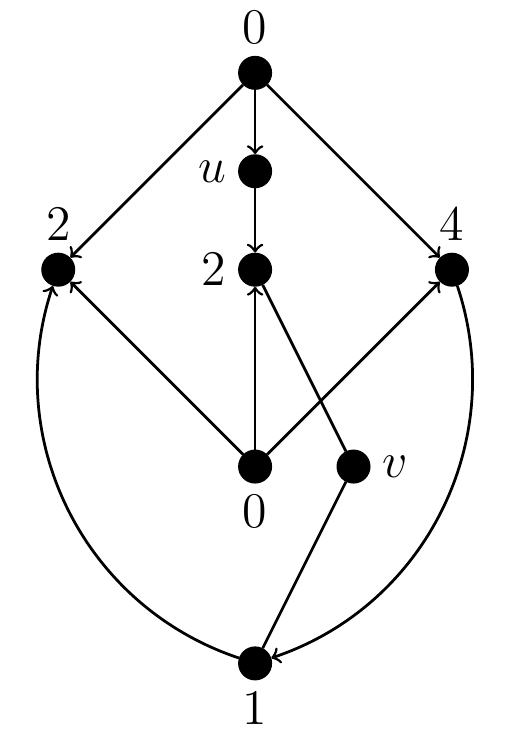} &
				\includegraphics[width =0.25\linewidth]{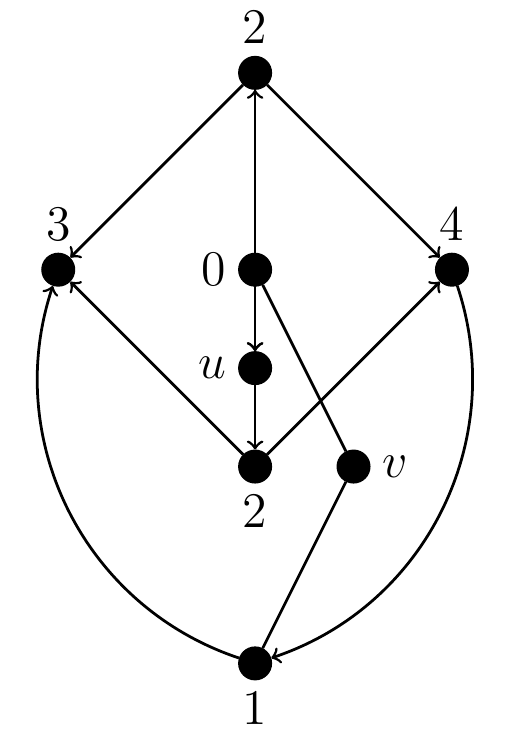} & 	
				\includegraphics[width =0.25\linewidth]{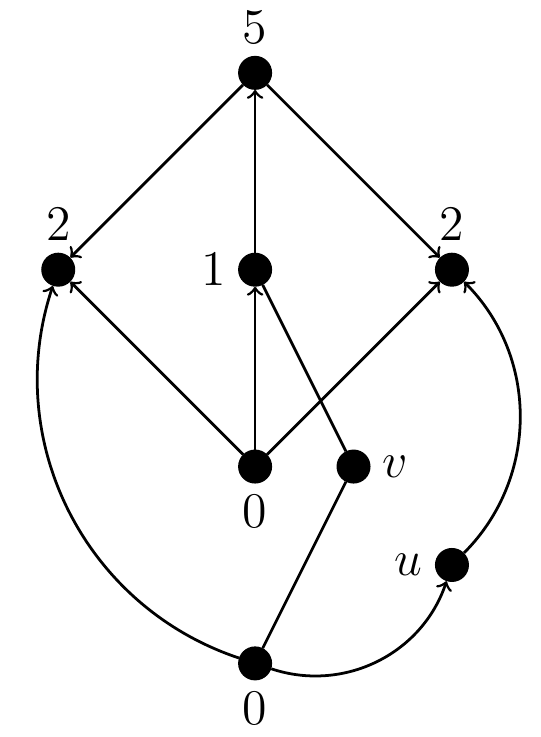}  &
				\includegraphics[width =0.25\linewidth]{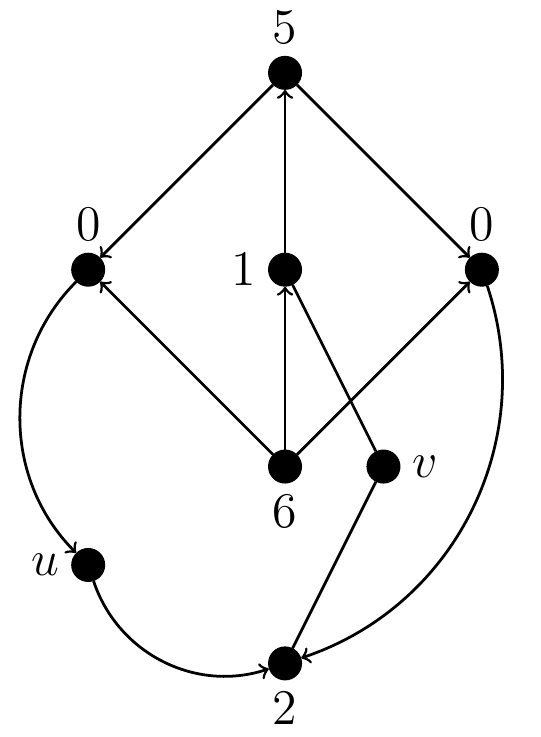} 
						
			\end{tabular}
		\end{center}
		\caption{Colourings for $G-z$ in Claim~\ref{claim:sourceSink} Case I when $H_u$ contains a copy of $Z_3$.}
		\label{orient:not2dp}
	\end{figure}

	As such it must be that the arc $u_2u_1$ corresponds to the one between $z_3$ and $z_1$ as shown in the first partially oriented graph in Figure~\ref{orient:not2dp}.
	We extend these colourings as shown in Figure~\ref{orient:not2dpAA}.
	Therefore if $H_u$ contains a subgraph  $Z \in \mathcal{Z}$, then there exists a homomorphism of $G-z$ to $QR_7$  in which $u$ and $v$ have different images, a contradiction. 
	This concludes \emph{Case I}.

\begin{figure}
	\begin{center}
		\begin{tabular}{cccc}
		\includegraphics[width = 0.25\linewidth]{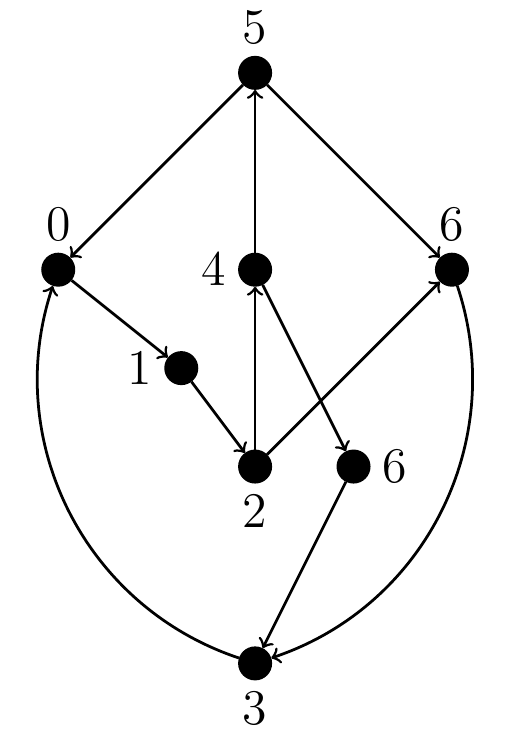} &
		\includegraphics[width = 0.25\linewidth]{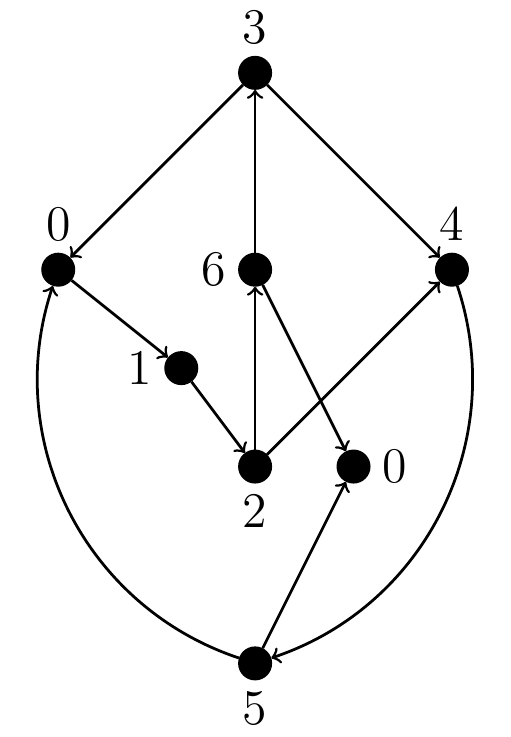}  &
			\includegraphics[width = 0.25\linewidth]{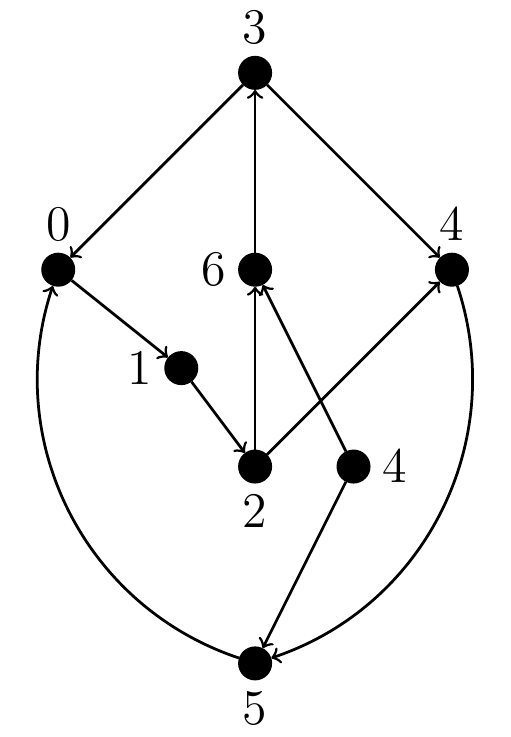} &
		\includegraphics[width = 0.25\linewidth]{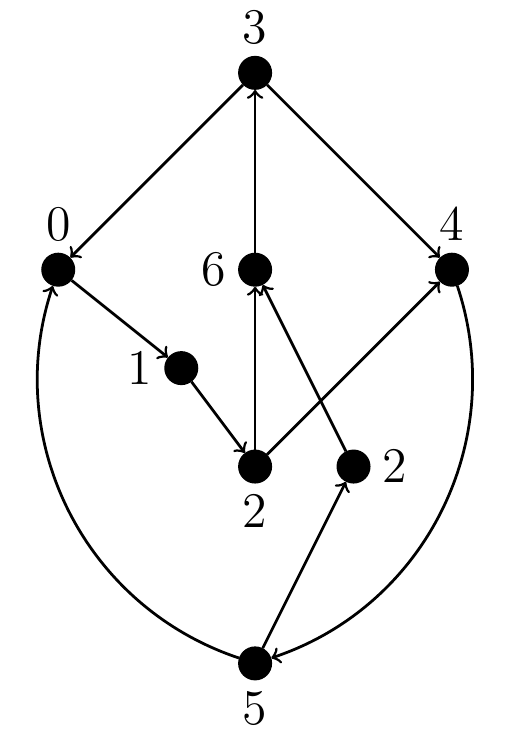}  
		\end{tabular}		
	\end{center}
			\caption{Colourings for $G-z$ in Claim~\ref{claim:sourceSink} Case I when $H_u$ contains a copy of $Z_3$ and $u_1u_2$ corresponds to $z_3z_1$.}
			\label{orient:not2dpAA}
\end{figure}

	\emph{Case II: $H_u$ contains a subgraph  $R \in \mathcal{R}$ and reduces to an oriented graph containing an element of $\mathcal{Z}$.}\\
	We derive a contradiction by constructing a homomorphism of $G-z$ to $QR_7$ in which $u$ and $v$ have different images.
	To do so we first show that $H_u$ contains a graph as in Figure~\ref{orient:singlereduce}.

	Consider $H_u^R$, the graph produced by reducing $H_u$.
	If $H_u^R$ contains a graph from $\mathcal{Z}$, then $H_u$ contains a graph as in Figure~\ref{orient:singlereduce}, as $H_u$ reduces to contain a graph from $\mathcal{Z}$ with a single reduction.
	If $H_u^R$ does not contain a graph from $\mathcal{Z}$, then it eventually reduces to a graph that contains some $Z \in \mathcal{Z}$.
	That is, $H_u^R$  contains a subgraph $S \in \mathcal{R}$.
	It cannot be that $S$ is contained in $H_u$ as otherwise it would be contained in $G-z$.
	Since this copy of $S$ is not contained in $H_u$, it must be that $S$ contains $r$, the vertex produced when reducing $H_u$.
	By Lemma~\ref{reducetoZ}, $H_u^R - r$ is reducible.	
	No vertex of $R$ is contained in $H_u^R$.
	Therefore every vertex and arc of  $S$ is contained in $G-z$,
 contradicting that $G-z$ is not reducible.
	Therefore $H_u$ contains a subgraph as in Figure~\ref{orient:singlereduce}, as required.

	Since the subgraph in Figure~\ref{orient:singlereduce} is not contained in $G-z$ but is contained in $H_u$, the arc $u_2u_1$ must appear in this subgraph.
	We now consider  to which arc $u_2u_1$ corresponds in the copy of the graph from Figure~\ref{orient:singlereduce}.
	If $u_1u_2$ corresponds to any arc other than the one between $x_2$ and $y_2$ or the one between $x_1$ and $y_1$, then  $G-z$ is reducible.
	Therefore  $u_2u_1$ corresponds to, without loss of generality, the arc between $x_2$ and $y_2$. 
	Observe the arcs incident with $x_5$ and $y_5$ that do not have their other ends at one of $x_3$, $x_4$, $y_3$ or $y_4$ in $G-z$   form a $2$-edge cut. 
	By Claim~\ref{claim:2edgecut}, these arcs must have a common endpoint of degree $2$.
	Since $G-z$ has only two vertices of degree $2$, this common endpoint must be $v$, as $u$ is the centre vertex on a $2$-dipath from $y_2$ to $x_2$.
	
	Since this graph must reduce to one that contains a copy of a graph from $\mathcal{Z}$, we may assume that neither $u_1$ nor $u_2$ are the centre of a $2$-dipath in $G -\{z,u\}$. 
	As otherwise, there is a reduction that does not give a copy of a graph from $\mathcal{Z}$.
	Therefore $G-z$ can be constructed from one of the four possible partial orientations given in Figure~\ref{orient:4orientations}. 
	
	In each of these cases, a homomorphism $\phi$ so that $\phi(u) \neq \phi(v)$ can be constructed as follows:
	\begin{itemize}
		\item $\phi(u_1) = \phi(y_1) = 0$;
		\item $\phi(u_2) = \phi(x_1) = 1$;
		\item $\phi(u) = \phi(v_2) =  4$;
		\item $\phi(v_1) = 6$.
	\end{itemize}

To complete the construction of the homomorphism in each case, we  define images for $y_3,y_4,x_3,x_4$ and $v$.
We do so as follows.
\begin{itemize}
\item If $y_4$ is an out-neighbour of $y_1$,  then let $\phi(y_3) = 1$ and $\phi(y_4) = 2$.
Otherwise, let  $\phi(y_3) = 5$ and $\phi(y_4) = 3$.
\item If $x_4$ is an out-neighbour of $x_1$, let $\phi(x_3) = 3$ and $\phi(x_4) = 5$.
Otherwise, let  $\phi(x_3) = 0$ and $\phi(x_4) = 4$.
\end{itemize}
Finally, to find an image of $v$ such that $\phi(v) \neq \phi(u)$ we apply Property~\ref{propty:pij} of $QR_7$.
Since $\phi(v_2) = 4$ it cannot be that $\phi(v) = \phi(u) = 4$.
		\begin{figure}
		\begin{center}
			\begin{tabular}{cc}
				\includegraphics{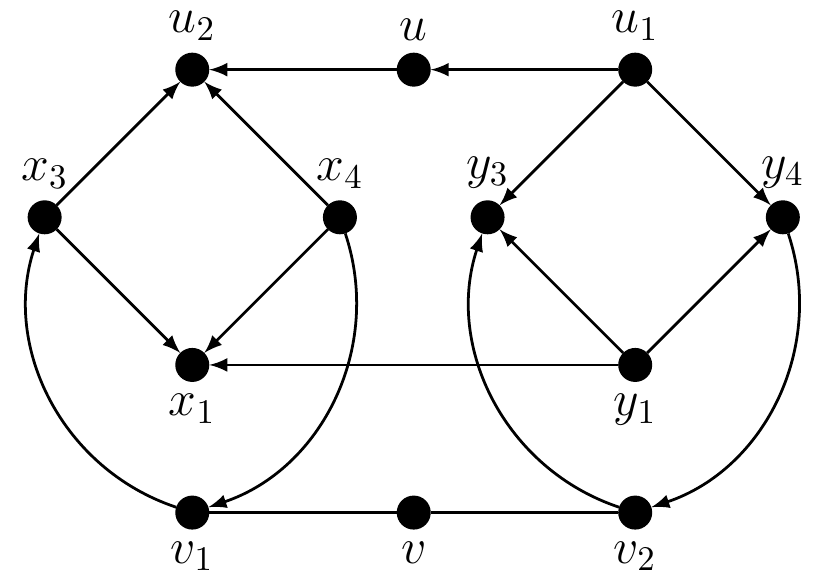} & \includegraphics{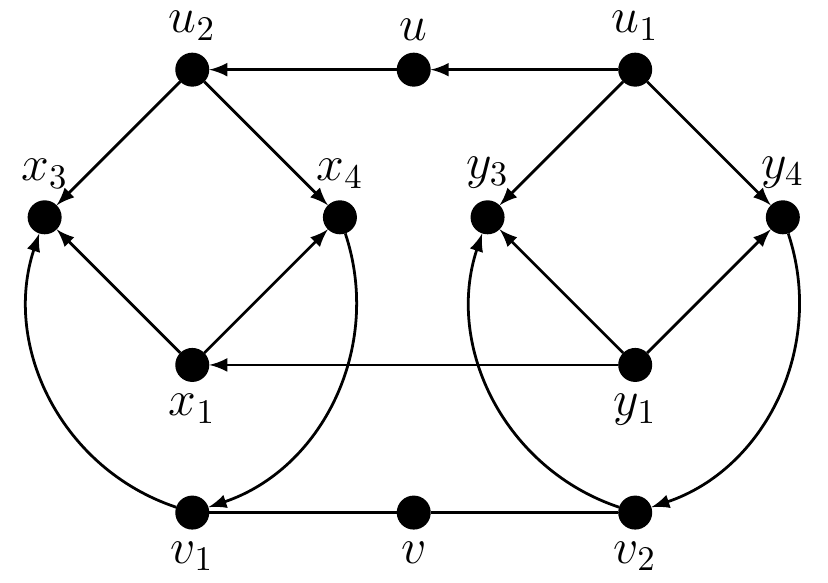} \\
				\\
				\includegraphics{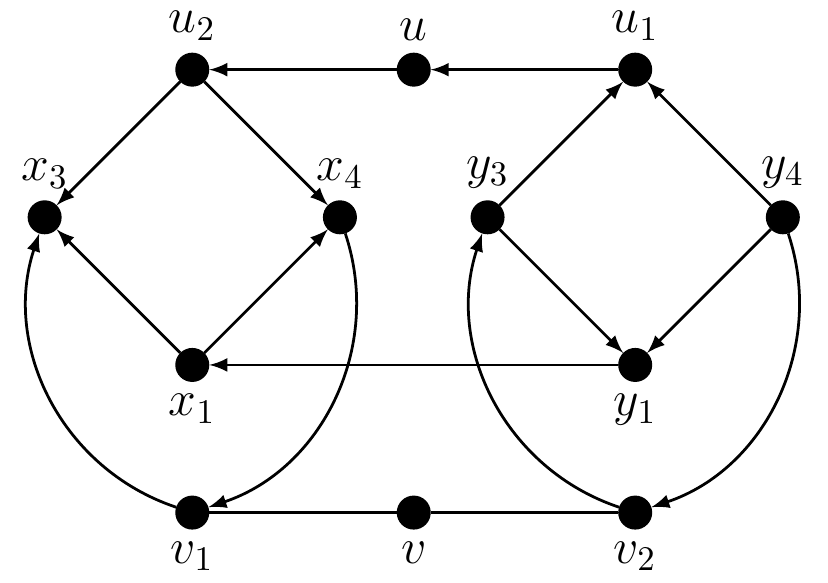} & \includegraphics{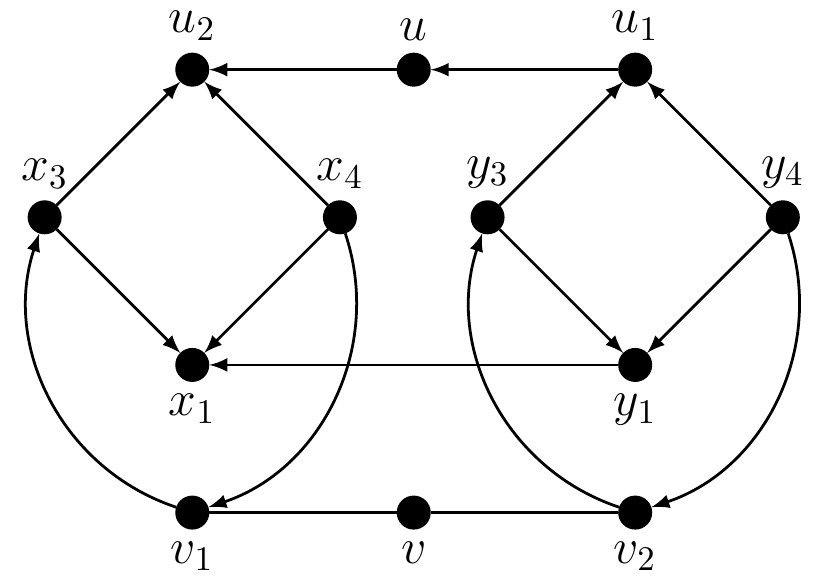} \\
			\end{tabular}
		\end{center}
		\caption{Possibilities for $G-z$ in Claim~\ref{claim:sourceSink} when $H_u$ is reducible. }
		\label{orient:4orientations}
	\end{figure}

	Therefore if $H_u$ contains a subgraph  $R \in \mathcal{R}$ and reduces to an oriented graph containing an element of $\mathcal{Z}$, then there exists a homomorphism of $G-z$ to $QR_7$  in which $u$ and $v$ have different images, a contradiction. 
	This concludes \emph{Case II}.
	
	Therefore each of $u$ and $v$ is either a source or sink vertex in $G -z$, as required.

		\begin{claim} \label{claim:3commclaim}
			$u_1$ and $u_2$ do not have three common neighbours in $G - z$.
		\end{claim}
	
	Suppose that $u_1$ and $u_2$ have three common neighbours: $u,x_1,x_2$. 
	Observe that $x_1$ and $x_2$ are not adjacent, as otherwise $uz$ is a cut arc in $G$.
	The arcs incident with $x_1$ and $x_2$ that are not incident with $u_1$ and $u_2$ have a common endpoint of degree $2$, by Claim~\ref{claim:2edgecut}. Since $G-z$ has no cut edge it must be that this common endpoint is $v$ by Claim~\ref{claim:deg2claim}. Therefore $x_1 = v_1$ and $x_2 = v_2$.
	By Claim \ref{claim:sourceSink} we may assume without loss of generality that $u$ is a source. We proceed in two cases and construct a homomorphism of $G-z$ to $QR_7$ where $u$ and $v$ are assigned different colours.
	
	\emph{Case I: $v$ is a source.}
	
	We construct $\phi: G-z \to QR_7$ as follows. Let $\phi(u) = 0$.
	
	If $u_1$ and $u_2$ are not the ends of a $2$-dipath then let $\phi(u_1) = \phi(u_2) = 1$.
	In this case, $v_1$ is either a common out-neighbour or a common in-neigbour of $u_1$ and $u_2$. 
	If $v_1$ is a common out-neighbour then let $\phi(v_1) = 5$, otherwise let $\phi(v_1) = 6$.
	This partial homomorphism can be extended to $v_2$ and $v$ using Property \ref{propty:pij} of $QR_7$.
	Observe that in every such extension we have $\phi(v) \neq \phi(u)$, as $0$ is not an in-neighbour of either $5$ or $6$, the possible images of $v_1$. 
 	
 	If $u_1$ and $u_2$ are the ends of a $2$-dipath, assume, without loss of generality, that $u_1v_1u_2$ is a $2$-dipath.
 	In this case, let $\phi(u_2)=1$ and $\phi(v_1) = 5$.
 	This partial homomorphism can be extended to $u_1,v_2$ and $v$ using Property \ref{propty:pij} of $QR_7$.
 	Observe that in every such extension we have $\phi(v) \neq \phi(u)$, as $0$ is not an in-neighbour of $5$, the image of $v_1$. 
	
	\emph{Case II: $v$ is a sink.}
	Let $\phi(u) = 0$.
	If $u_1$ and $u_2$ are not the ends of a $2$-dipath then let $\phi(u_1) = \phi(u_2) = 1$.
	In this case, $v_1$ is either a common out-neighbour or a common in-neigbour of $u_1$ and $u_2$. 
	If $v_1$ is a common out-neighbour, then let $\phi(v_1) = 2$, otherwise let $\phi(v_1) = 4$.
	This partial homomorphism can be extended to $v_2$ and $v$ using Property \ref{propty:pij} of $QR_7$.
	Observe that in every such extension we have $\phi(v) \neq \phi(u)$, as $0$ is not an out-neighbour or either $1$ or $4$, the possible images of $v_1$. 
	
	If $u_1$ and $u_2$ are the ends of a $2$-dipath, assume, without loss of generality, that $u_1v_1u_2$ is a $2$-dipath.
	In this case, let $\phi(u_2)=1$ and $\phi(v_1) = 2$.
	This partial homomorphism can be extended to $u_1,v_2$ and $v$ using Property \ref{propty:pij} of $QR_7$.
	Observe that in every such extension we have $\phi(v) \neq \phi(u)$, as $0$ is not an out-neighbour of  $2$, the image of $v_1$.

	\begin{claim}  \label{claim:notTwo}
		$|\{u_1,u_2,v_1,v_2\}| \neq 2$.
	\end{claim}
	Suppose the contrary.
	By Claim~\ref{claim:indep}, $u_1$ and $u_2$	are not adjacent.
	By Claim~\ref{claim:deg2claim} each of $u_1$ and $u_2$ has degree $3$ in $G$.
	Therefore the arcs incident with $u_1$ and $u_2$ that  do not have  $u$ or $v$ as an endpoint form a  $2$-edge cut in $G$.
	By Claim~\ref{claim:cutclaim}, these arcs have a common endpoint.
	However, this implies that $u_1$ and $u_2$ have three common neighbours, contradicting Claim~\ref{claim:3commclaim}.
	Therefore $|\{u_1,u_2,v_1,v_2\}| \neq 2$.
	
	\begin{claim} \label{claim:notThree}
		$|\{u_1,u_2,v_1,v_2\}| \neq 3$.
	\end{claim}
	
	Suppose the contrary. 
	Assume, without loss of generality, that $u_1 = v_1$. 
	Since $u$ and $v$ have the same image in any homomorphism of $G - z$ to $QR_7$, it must be that the arc between $u_1$ and $u$ and the arc between $u_1$ and $v$  have the same orientation with respect to $u_1$. That is, these arcs either both have their head at $u_1$ or both have their tail at $u_1$. Otherwise, $u$ and $v$ would be at the ends of a $2$-dipath.

	Consider the properly subcubic graph $A_u$ constructed by removing $z$ and $v$ and adding an arc between $u$ and $v_2$ and orienting it so that $\{u_1,u,v_2\}$ induces a $2$-dipath.
	If $A_u$ admits a homomorphism to $QR_7$, then observe that it may be extended to include $v$ by applying  Property~\ref{propty:pij} of $QR_7$. 
	In such an extension, $u$ and $v$ have different images as $uv_2v$ form a $2$-dipath in the graph constructed by adding $v$ to $A_u$. 
	The existence of such a homomorphism is a violation of the assumption that $G$ does not admit a homomorphism to $QR_7$. 
	Therefore it must be that $A_u$ either contains a copy of a graph from $\mathcal{Z}$ or eventually reduces to a graph that contains some $Z \in \mathcal{Z}$.
	In each case we derive a contradiction by constructing a homomorphism of $G-z$ to $QR_7$ so that $u$ and $v$ have different images.
	Observe that $u_1$ has degree $2$ in $A_u$.

	\emph{Case I: $A_u$  contains a subgraph  $Z \in \mathcal{Z}$.}\\
	Assume $uu_1$ is an arc.
	By Claim \ref{claim:sourceSink}, we have  $uu_2$, $vu_1$ and $vv_2$ are arcs.
	Suppose   $A_u$  contains a subgraph  $Z \in \mathcal{Z}$. Since adding the arc between $v_2$ and $u$ created this $Z$, it must be that this arc appears in the copy of $Z$. 
	If $u_1=v_1$ appears in this copy of $Z$, then it must correspond to a vertex of degree $2$ in $Z$, as $u_1$ has degree $2$ in $A_u$.
	As such, $u$ corresponds to $z_4$.
	This implies that $u_2$ and $v_2$ correspond to $z_1$ and $z_2$ in some order.
	Since $v$ does not appear in $H_u$, if $Z = Z_3$, then $z_6$ has degree three in $G-z$.
	However, if this is the case we notice that an arc incident with $z_6$ is a cut arc in $G-z$.
	Therefore $G - z$ is the oriented graph produced by oriented the graph in Figure~\ref{orient:claim6II}.
	A homomorphism of this graph to $QR_7$ can be constructing by letting $\phi(v_2) = \phi(u) = 0$, $\phi(u_1) = 1$, $\phi(v) = 6$ and finding images for the remaining vertices using Property \ref{propty:pij} of $QR_7$.
	 	\begin{figure}
	\begin{center}
		\includegraphics{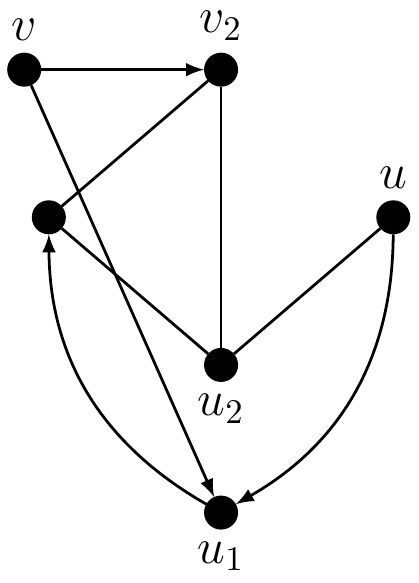} \quad \quad \includegraphics{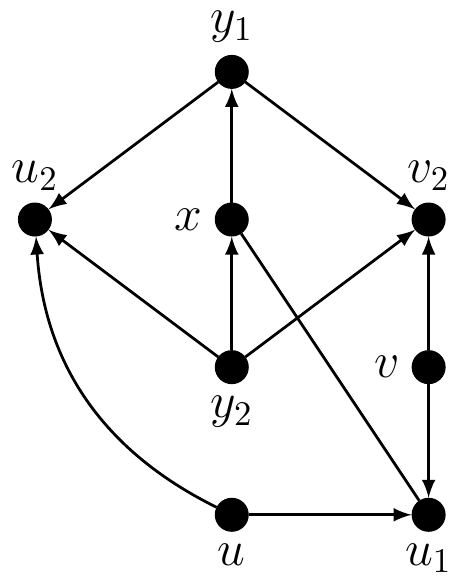}
	\end{center}
	\caption{{$G-z$ in Case \emph{I} of Claim $8$ when $u_1$ does and does not appear respectively  in $Z$. }}
	\label{orient:claim6II}
\end{figure}

	The existence of such a homomorphism is a contradiction, therefore $u_1$ does not correspond to any vertex in $Z$.
	Therefore $u$ corresponds to vertex of degree $2$ in $Z$.
	Without loss of generality, we may assume $u$ corresponds to $z_5$.
	Since the arc between $u$ and $v_2$ is contained in $A_u$, it must be that $v_2$ is contained in $A_u$.
	Therefore $v_2$ corresponds to $z_4$ and $u_2$ corresponds to $z_3$.
	
	If $z_6$ does not exist we see that the arc incident with $u_1$ that does not have its endpoint at either $u$ or $v$ is a cut arc in $G-z$, a contradiction of Claim~\ref{claim:cutclaim}.
	Thus Z = $Z_3$.
	By hypothesis, $z_6$ does not correspond to $u_1$ in $H_u$.
	If $u_1$ and the vertex corresponding to $z_6$, say $x$ are not adjacent, then the arc incident with $x$ that is not contained in $Z$ and the arc incident with $u_1$ that does not have its endpoint at either $u$ or $v$ form a $2$-edge cut in $G-z$.
	By Claim~\ref{claim:2edgecut}, these two arcs have a common endpoint of degree $2$ in $G-z$.
	The existence of such a vertex is a violation of Claim~\ref{claim:deg2claim}.
	Therefore $u_1$ and $z_6$ are adjacent.
	
	The graph $G-z$ is configured as in Figure \ref{orient:claim6II}.
	We find a homomorphism to $QR_7$ as follows: $\phi(u_1)=\phi(u_2) = \phi(v_1) = 0$;  $\phi(y_1) = 5$, $\phi(y_2) = 6$.
	If the arc between $x$ and $u_1$ has its head at $u$, let $\phi(x) = 3$.
	Otherwise let $\phi(x) = 1$.
	
	This homomorphism can be extended to so that $\phi(u) \neq \phi(v)$
	The existence of this homomorphism contradicts that $A_u$ contains a copy of a graph from $\mathcal{Z}$.
	This completes Case I.

	\emph{Case II: $A_u$ contains a copy of a graph $R \in \mathcal{R}$ and reduces to a graph that contains an element of $\mathcal{Z}$} as a subgraph.\\
	Suppose that $A_u$ contains a copy of a graph $R \in \mathcal{R}$. 
	We claim $A_u$ contains a copy of the graph in Figure~\ref{orient:singlereduce}.
	Consider $A_u^R$, the graph produced by reducing $A_u$.
	If $A_u^R$ contains a graph from $\mathcal{Z}$, then $A_u$ contains a graph as in Figure~\ref{orient:singlereduce}, as $A_u$ reduces to contain a graph from $\mathcal{Z}$ with a single reduction.
	If $A_u^R$ does not contain a graph from $\mathcal{Z}$, then it is reducible.
	That is, $A_u^R$  contains a subgraph $S \in \mathcal{R}$.
	Since this copy of $S$ is not contained in $A_u$, it must be that $S$ contains $r$, the vertex produced when reducing $A_u$.
	By Lemma~\ref{reducetoZ}, $A_u^R - r$ is reducible.	
	No vertex of $R$ is contained in $A_u^R$.
	Therefore every vertex and arc of  $S$ is contained in $G-z$,  contradiction as $G-z$ is not reducible.
	Therefore $A_u$ contains a subgraph as in Figure~\ref{orient:singlereduce}. 
	
	Since $G-z$ is reduced it must be that $u$ corresponds to either $x_1, x_2,y_1$ or $y_2$, as otherwise $G-z$ would be reducible. 
	Without loss of generality, assume $u$ corresponds $x_1$.
	If $u_1$ and $u_2$ correspond to $x_3$ and $x_4$, then  $u_1$ and $u_2$ have three common neighbours, a contradiction of Claim~\ref{claim:3commclaim}. 
	Without loss of generality, assume $u_1$ corresponds to $x_3$ and $u_2$ corresponds to $y_1$.
	Therefore $v_2$ corresponds to $x_4$.
	However, here we see that $G-z$ is reducible, a contradiction.
	This implies that $A_u$ does not reduce to a graph that contain an element of $\mathcal{Z}$ as a subgraph, a contradiction.

	By \emph{Case I} and \emph{Case II}, we have  $|\{u_1,u_2,v_1,v_2\}| \neq 3$.

	\begin{claim} \label{claim:notFour}
		$|\{u_1,u_2,v_1,v_2\}| \neq 4$.
	\end{claim}
	
	Suppose $|\{u_1,u_2,v_1,v_2\}| = 4$. Let $A_{v_1}$ be the oriented graph constructed from $G$ by removing $z$ and $v$ and adding the edge between $v_1$ and $u$, and then orienting it so that this arc has its head at $v_1$ if and only if the arc between $v_1$ and $v$ has its tail at $v_1$.
	
	If $A_{v_1}$ admits a homomorphism to $QR_7$, then by Property~\ref{propty:pij} of $QR_7$ and Claim~\ref{claim:sourceSink} we can extend this homomorphism to include $v$. 
	However in this case it cannot be that $u$ and $v$ have the same image; there is a $2$-dipath between them. 
	Therefore $A_{v_1}$ does not admit a homomorphism to $QR_7$. As such it either contains a copy of an oriented graph in $\mathcal{Z}$ or eventually reduces to a graph that contains some $Z \in \mathcal{Z}$. 
	Similarly we construct $A_{v_2}$ and assert that it contains a copy of an oriented graph in $\mathcal{Z}$ or eventually reduces to a graph that contains some $Z \in \mathcal{Z}$. 
	
	We claim that neither of $A_{v_2}$  and $A_{v_1}$ contain a graph from $\mathcal{R}$. Suppose  $A_{v_1}$ contains a graph $R \in \mathcal{R}$.  
	Using the argument in Case II in the proof of Claim~\ref{claim:notThree}, we claim $A_{v_1}$ contains a copy of the graph in Figure~\ref{orient:singlereduce}.
 	Since $G$ is not reducible it must be that $u$ corresponds to either $x_1,x_2,y_1$ or $y_2$.
 	Without loss of generality,  assume $u$ corresponds to $x_2$ and $v_1$ corresponds to $y_2$. 
 	However, if this is true, $u_1$ and $u_2$ have three common neighbours in $G- z$, contradicting Claim~\ref{claim:3commclaim}.
	
	Therefore each of $A_{v_2}$  and $A_{v_1}$ contains a graph from $\mathcal{Z}$. 
	Assume that $A_{v_1}$ contains  $Z \in \mathcal{Z}$ and  $A_{v_2}$ contains  $Z^\prime \in \mathcal{Z}$. 
	Observe that since $G-z$ contains no graph from $\mathcal{Z}$ it must be that the arc between $u$ and $v_1$ (respectively $v_2$) is contained in $Z$ (respectively $Z^\prime$).
	We first show that each of $u_1$ and $u_2$ has at least two common neighbours with one of $v_1$ and $v_2$ in $G$. 
	We do this by considering the degree of the vertex to which $u$ corresponds  in $Z$. 
	
	If $u$ corresponds to a vertex of degree $2$ in $Z$ it must be that  $v_1$ corresponds to a vertex of degree $3$, as there is no pair of adjacent degree $2$ vertices in any member of $\mathcal{Z}$.
	If $u$ corresponds to $z_5$, then one of $u_1$ and $u_2$ corresponds to $z_4$.  Therefore at least one of $u_1$ and $u_2$ has two common neighbours with $v_1$  in $G$.
	
	If $u$ corresponds to a vertex of degree $3$ in $Z$, then it cannot be that $v_1$ corresponds to a vertex of degree $2$, as otherwise $u_1$ and $u_2$ would have three common neighbours. 
	Therefore, if $u$ corresponds to $z_1$, then we may assume, without loss of generality, that $v_1$ corresponds to $z_3$.
	In this case, either  $u_1$ or $u_2$ corresponds to $z_4$, which has two common neighbours with $v_1$ in $G$ .

	By considering a similar argument for $A_{v_2}$ we see that at least one of $u_1$ and $u_2$ has two common neighbours with $v_2$  in $G$.
	However, since $v_1$ and $v_2$ have at most two common neighbours and since $u_1$ is not adjacent to $v$ (one of their common neighbours), we conclude that $u_1$ does not have two common neighbours with both $v_1$ and $v_2$.
	Therefore each of $u_1$ and $u_2$ has at least two common neighbours with one of $v_1$ and $v_2$ in $G$. 
	Without loss of generality, assume  $u_1$ has two common neighbours with $v_1$ and $u_2$ and has two common neighbours with $v_2$. 
	Let $x_1$ and $x_2$ be the common neighbours of $u_1$ and $v_1$.
	And let  $x_1^\prime$ and $x_2^\prime$ be the common neighbours of $u_2$ and $u_1$.

	We claim the vertices of $G-z$ are configured as in one of the two partially oriented graphs in Figure~\ref{claim:notFourfig} where the unoriented edges are oriented so that each of $A_{v_2}$ and $A_{v_2}$ contains an element from $\mathcal{Z}$.
	By the previous argument, $u_1$ and $v_1$ (respectively $u_2$ and $v_2$) have a pair of common neighbours. Since $A_{v_1}$ (respectively $A_{v_2}$) contains a copy of  $ Z \in \mathcal{Z}$ (respectively $Z^\prime \in \mathcal{Z}$), these common neighbours must either be adjacent or are the ends of a $2$-dipath. 
	The former case corresponds to $z_6$ (respectively $z_6^\prime$) not existing in $Z$ (respectively $Z^\prime$), whereas the latter case corresponds to $z_6$ (respectively$z_6^\prime$) existing in $Z$ (respectively $Z^\prime$). It remains to show that both $z_6$ and $z_6^\prime$ are present or both not present in $Z$ and $Z^\prime$ respectively, as shown in Figure~\ref{claim:notFourfig}. 
	
	Without loss of generality, we may assume $uu_1$ is an arc.
	By Claim~\ref{claim:sourceSink}, we have that $uu_2$ is an arc.
	Since $G-z$ is not reducible, it must be that $vv_1$ is an arc.
	By Claim~\ref{claim:sourceSink}, we have that $vv_2$ is an arc.
	
	If  $Z$ and $Z^\prime$ both contain a vertex corresponding to $z_6$, these vertices must have degree $3$ by Claim~\ref{claim:deg2claim}. Let $y$ and $y^\prime$ be these vertices, as shown in Figure~\ref{claim:notFourfig}.
	If these vertices are not adjacent, then, the arcs incident with $y$ and $y^\prime$ that do not appear in $Z$ and $Z^\prime$, respectively, have a common endpoint of degree $2$, a contradiction of Claim~\ref{claim:deg2claim}. 
	Therefore these vertices must be adjacent.
	Assume, without loss of generality, that this arc is oriented as shown in Figure~\ref{claim:notFourfig}.
	
	If, without loss of generality, only $Z$ contains a vertex corresponding to $z_6$, then such a vertex must have degree $3$ by Claim~\ref{claim:deg2claim}.
	However, in this case, the arc incident with this vertex that is not contained in $Z$ is a cut arc, a contradiction of Claim~\ref{claim:cutclaim}.
	
	To derive a contradiction we find a homomorphism to $QR_7$ for each of the possible orientations of the graphs arising from the partial orientations in Figure~\ref{claim:notFourfig}, where the unoriented edges take orientations so that $Z$ and $Z^\prime$ exist.
	
	Each of $A_{u_1}$ and $A_{u_2}$ contain a copy of a graph from $\mathcal{Z}$.
	Therefore  the edge between $x_1$ and $u_1$ has the same orientation (with respect to $x_1$) as the one between $x_1$ and $v_1$.
	A similar statement holds for the following pairs of edges: 
		\begin{itemize}
			\item $x_2u_1$ and $x_2u_2$ (with respect to $x_2$);
			\item $x_1^\prime u_1$ and $x_1^\prime u_2$ (with respect to $x_1^\prime$); and
			\item $x_2^\prime u_1$ and $x_2^\prime u_2$ (with respect to $x_2^\prime$).
		\end{itemize}
	
	To find such homomorphisms, let $\phi(u) = 0, \phi(v) = 1$ and $\phi(u_1)=\phi(v_1) = \phi(u_2) = \phi(v_2) = 2$.

	We first construct a homomorphism for the case that $z_6$ (i.e., $y$) and $z_6^\prime$ (i.e., $y^\prime$) do not exist.
	If $x_1u_1x_2$ is a $2$-dipath, then so is $x_1v_1x_2$, as $Z \in \mathcal{Z}$.
	In this case, let $\phi(x_1)=0$ and $\phi(x_2) = 4$.
	If $x_1u_1$ is an arc, then so are $x_1v_1, x_2u_1$ and $x_2v_1$, as $Z \in \mathcal{Z}$.
	In this case, let $\phi(x_1)=0$ and $\phi(x_2) = 1$.
	If $u_1x_1$ is an arc, then so are $v_1x_1, u_1x_2$ and $v_1x_2$, as $Z \in \mathcal{Z}$.
	In this case, let $\phi(x_1)=3$ and $\phi(x_2) = 5$.
	A similar argument gives an image of each vertex $x_1^\prime$ and $x_2^\prime$ based on the orientation of the edges with an end at $u_2$.
	This defines a homomorphism to $QR_7$.
	
	We now construct a homomorphism for the case that $z_6$ (i.e., $y$) and $z_6^\prime$ (i.e., $y^\prime$) do  exist.

	Let $\phi(y) =1$ and $\phi(z_6^\prime) = 0$.
	
	If $x_1u_1$ is an arc, then so are $x_1v_1, x_2u_1$ and $x_2v_1$, as $Z \in \mathcal{Z}$.
	In this case, let $\phi(x_1)=0$ and $\phi(x_2) = 5$.
	If $u_1x_1$ is an arc, then so are $v_1x_1, u_1x_2$ and $v_1x_2$, as $Z \in \mathcal{Z}$.
	In this case, let $\phi(x_1)=6$ and $\phi(x_2) = 3$.

	If $x_1^\prime u_2$ is an arc, then so are $x_1^\prime v_2, x_2^\prime u_2$ and $x_2^\prime v_2$, as $Z \in \mathcal{Z}$.
	In this case, let $\phi(x_1^\prime )=5$ and $\phi(x_2^\prime ) = 1$.
	If $u_2x_1^\prime $ is an arc, then so are $v_2x_1^\prime , u_2x_2^\prime $ and $v_2x_2^\prime $, as $Z \in \mathcal{Z}$.
	In this case, let $\phi(x_1^\prime )=3$ and $\phi(x_2^\prime ) = 4$.
		This defines a homomorphism to $QR_7$.
		
	Therefore if each of $A_{v_2}$  and $A_{v_1}$ contains a graph from $\mathcal{Z}$, then $G$ admits a homomorphism to $QR_7$, a contradiction.
	Thus, at least one of $A_{v_2}$  and $A_{v_1}$ admits a homomorphism to $QR_7$.
	However, this implies that there exists $\phi: G-z \to QR_7$ in which $\phi(u)\neq \phi(v)$, a contradiction.
	Therefore $|\{u_1,u_2,v_1,v_2\}| \neq 4$.
	
	\begin{figure}
		\begin{center}
			\includegraphics[width = \linewidth]{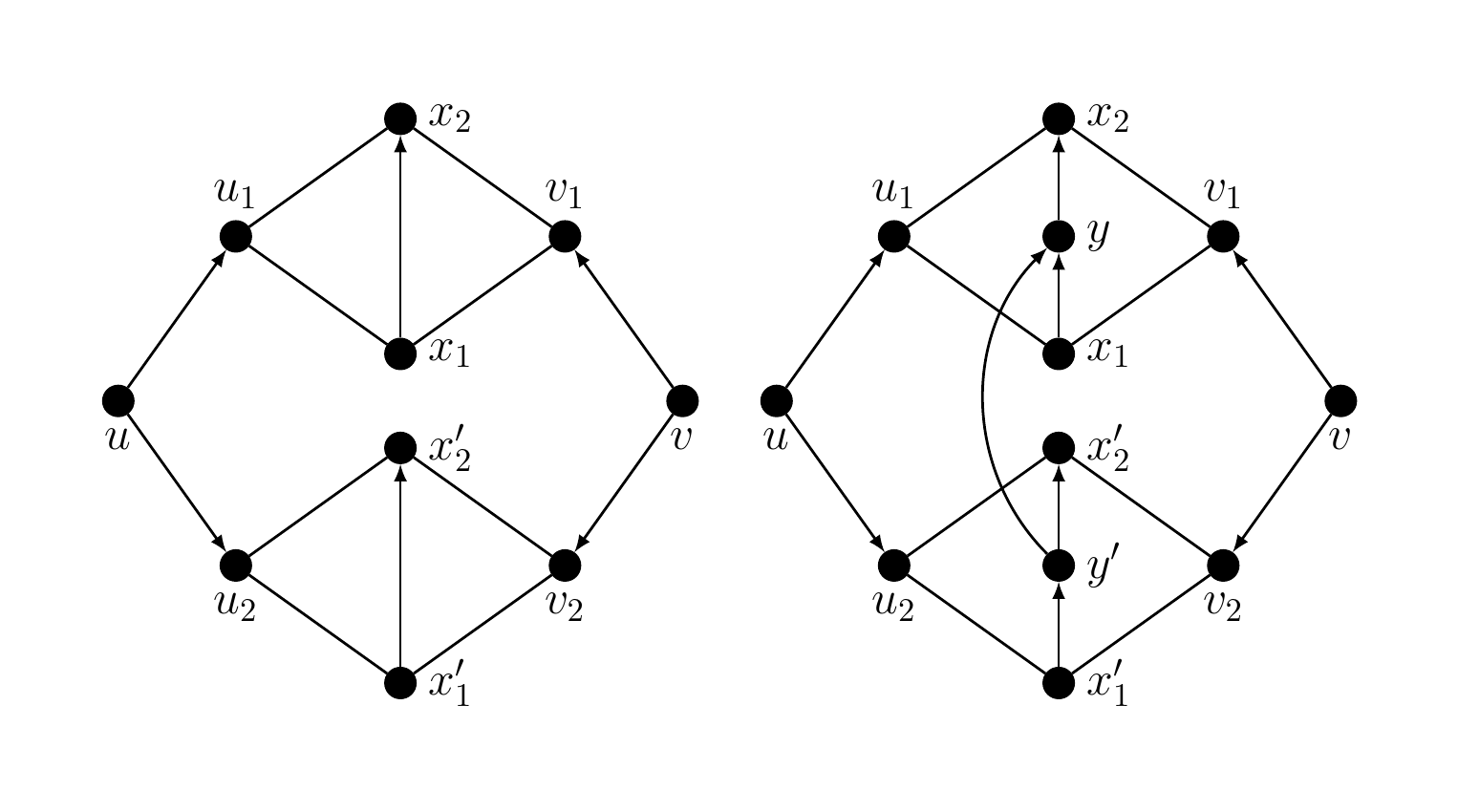}
		\end{center}
		\caption{The possibilities for $G-z$ when $|\{u_1,u_2,v_1,v_2\}| = 4$.}
		\label{claim:notFourfig}
	\end{figure}
	
	Together Claims~\ref{claim:notTwo},\ref{claim:notThree} and~\ref{claim:notFour} contradicts the existence of $G$,  a minimum counter-example. This completes the proof.
\end{proof}

\begin{thm}\label{thm:noadjst}
	If $G$ is a properly subcubic graph with no vertex of out-degree three adjacent to a vertex of in-degree three, then $G$ admits a homomorphism to $QR_7$.
\end{thm}

\begin{proof}
	Let $G$ be properly subcubic graph with no vertex of out-degree three adjacent to a vertex of in-degree three.
	Observe that $G^R$ has this same property.
	 Since every element of $\mathcal{Z}$ contains a vertex with out-degree three adjacent to a vertex with in-degree three, $G^R$ necessarily contains no subgraph from $\mathcal{Z}$. The result follows by The Reduction Lemma and Lemma \ref{lem:heavyLift}.
\end{proof}

Lemma \ref{lem:heavyLift} allows us to bound the oriented chromatic number of $\mathcal{F}_3$.

\begin{thm} \label{orient:9colour}
	Any orientation  $G$ of a connected graph with $\Delta \leq 3$, then $\chi_o(G) \leq 9$
\end{thm}

\begin{proof}
	Let $G$ be an oriented connected cubic graph. 
	If $G$ contains no source or no sink, then removing a single arc $uv$ from $G$ yields an oriented graph which satisfies the hypothesis of Lemma \ref{thm:noadjst}. 
	Let $\phi: G-\{uv\} \to QR_7$ be a homomorphism.
	We extend $\phi$ to be a homomorphism of $G$ to $QR_7$, by letting $\phi(u) = 7$ and $\phi(v) = 8$.
	Otherwise assume, without loss of generality  $G$ contains a source vertex. Let $S$ be the set of source vertices in $G$. 
	Observe that $G-S$ satisfies the hypothesis of Theorem  \ref{thm:noadjst}. 
	Let $\phi: G-S \to QR_7$ be a homomorphism.
	We extend $\phi$ to be a homomorphism of $G$ to $QR_7$, by letting $\phi(s) = 7$ for all $s \in S$.
\end{proof}

\begin{cor}
If $G$ is an oriented cubic graph that contains a source or a sink, then $\chi_o(G) \leq 8$.
\end{cor}

\section{Oriented Colourings of Graphs with Maximum Degree Four}\label{deg4}
Let  $\mathcal{F}_4$ be the family of orientations of connected graphs with maximum degree $4$.
As with our improved bound for orientations of connected cubic graphs, we use a non-zero quadratic residue tournament  as a means to construct a target on no more than $69$  vertices for members of this family. 
This construction gives an upper bound for the oriented chromatic number of the family of orientations of graphs with maximum degree $4$.

We begin by observing two properties of the Paley tournament on $67$ vertices.  

\begin{prop} [\cite{GS71}]
	$QR_{67}$ has Property $P_{4,1}$ and Property $P_{3,2}$.
\end{prop} 

\begin{prop}
	$QR_{67}$ is vertex transitive, arc transitive and self-converse.
\end{prop}

\begin{lem} \label{orient:QR67}
	Every orientation of a connected properly subquartic graph admits a homomorphism to $QR_{67}$.
\end{lem}

\begin{proof}
	Let $G$ be a minimum counter-example with respect to number of vertices and subject to that with respect to the number of arcs. 
	We consider cases based on the minimum degree of a vertex in $G$. 
	Let $z$ be a vertex of minimum degree in $G$. 
	Since $G$ is properly subquartic, it must be that $z$ has degree $1$, $2$ or $3$.
	In each case we derive a contradiction.
	Since $G$ is a minimum counter-example and $QR_{67}$ is vertex and arc transitive, we observe that $G$ contains no cut vertex or cut edge.
	By the minimality of $G$, there is a homomorphism of $QR_{67}$ to $G-z$.
	
	\emph{Case I: $z$  has degree $1$.}\\ 
	If $z$ has degree $1$, then the arc incident with $z$ is a cut arc, a contradiction, as $G$ contains no cut vertex.
		
	\emph{Case II: $z$ has degree $2$.}\\ 
	Let $u$ and $v$ be the neighbours of $z$ in $G$. 
	If both $u$ and $v$ have $z$ as an out-neighbour (respectively in-neighbour), then any homomorphism of $G-z$ to $QR_{67}$ can be extended to include $z$, since $QR_{67}$ has Property $P_{2,1}$. 
	Thus, without loss of generality, we may assume that $uz, zv \in E(G)$. 
	We may further assume that in every homomorphism of $G-z$ to $QR_{67}$, $u$ and $v$ have the same image, as otherwise any homomorphism of $G-z$ to $QR_{67}$ can be extended to include $z$ by using Property $P_{2,1}$ of $QR_{67}$.
	
	Consider a homomorphism $\phi: G -\{u,z\} \to QR_{67}$. 
	Using Property $P_{3,2}$, $\phi$ can be extended in at least two ways to include $u$. 
	In particular $\phi$ can be extended in such a way that $u$ and $v$ have a different image, a contradiction as we assumed  in every such homomorphism  $u$ and $v$ have the same image.

	\emph{Case III: $z$ has degree $3$.}\\ 
	Let $u,v,w$ be the neighbours of $z$ in $G$. 
	Following \emph{Case II} we may assume that there is a $2$-dipath with centre vertex $z$ in $G$.
	Since $QR_{67}$ is self converse, we assume, without loss of generality, that $uz,zv,zw \in E(G)$.
	Further, following \emph{Case II} we assume in every homomorphism of  $G-z$ to $QR_{67}$ that $u$ has the same image of at least one of $v$ and $w$. 
	From this it follows that $u$ is adjacent to at most one of $v$ and $w$.
	We proceed based on the existence of arcs between $u,v$ and $w$.
	
	\emph{Subcase III.i: $u,v,w$ form an independent set.}\\
	Consider a homomorphism $\phi:  G -\{u,v,w,z\} \to QR_{67}$. 
	We extend $\phi$ to be a homomorphism of $QR_{67}$ to $G -z$ so that $u$ does not have the same image as either $v$ or $w$,  thus deriving a contradiction.
	Since $QR_{67}$ has Property $P_{3,2}$, there are at least two choices for each of $u,v,w$ when extending $\phi$. 
	Since these vertices form an independent set, the chosen image of either $u,v$, or $w$ in an extension of $\phi$ does not affect the list of choices for the other two vertices.
	By hypothesis, no matter how these choices are made, it must be that $u$ has the same image as at least one of $v$ or $w$. 
	Consider a graph with vertex set $\{u^\prime, v^\prime, w^\prime\}$ and edge set $\{u^\prime v^\prime, u^\prime w^\prime\}$. 
	If we assign to $u^\prime$ (respectively $v^\prime$ and $w^\prime$) the same list of vertices of $QR_{67}$ that are available for $u$ (respectively $v$ and $w$) when extending $\phi$, then a list colouring of this constructed graph corresponds exactly to an extension of $\phi$ to include $u,v$ and $w$ where $u$ does not have the same image either as $v$ or $w$. 
	Since $QR_{67}$ has Property $P_{3,2}$ each of these lists has cardinality at least two. 
	Since this constructed graph, $K_{2,1}$, is $2$-list colourable, such an extension must exist, a contradiction as we assumed that in every such homomorphism  $u$ has the same image of at least one of $v$ and $w$.
	
	\emph{Subcase III.ii: $u,v,w$ do not form an independent set}. \\
	Since $u$ is adjacent to at most one of $v$ and $w$, we assume, without loss of generality, that there is an arc between $u$ and $v$ (in some direction), but no arc between $u$ and $w$.
	Consider a homomorphism $\phi: G -\{w,z\}\to QR_{67}$. 
	Since $QR_{67}$ has Property $P_{3,2}$, $\alpha$ can be extended to include $w$ in at least two ways. 
	In particular, $\phi$ can be extended to include $w$ so that $\phi(w) \neq \phi(u)$.
	Since $u$ and $v$ are adjacent, there is a homomorphism of $G-z$ to $QR_{67}$ in which $u$ does not have the same image as $v$ or $w$, a contradiction.
\end{proof}

\begin{thm} \label{orient:degree4}
	For the family, $\mathcal{F}_4$, of orientations of connected graphs with maximum degree at most four,  $11 \leq \chi_o(\mathcal{F}_4) \leq 69$.\end{thm}

\begin{proof}
	Figure~\ref{7clique} gives a member of  $\mathcal{F}_4$ that is an oriented clique on $11$ vertices.
	Therefore  $\chi_o(\mathcal{F}_4)\geq 11$.

	Let $G$ be an oriented graph with maximum degree $4$. 
	Consider $uv \in E(G)$.
	By Lemma~\ref{orient:degree4}, there exists a homomorphism  $\phi: G - \{uv\} \to QR_{67}$.
	We extend $\phi$ to be homomorphism of $G$ to a target on no more than $69$ vertices by letting $\phi(u) = 67$ and $\phi(v) = 68$.	
\end{proof}

\section{Future Directions and Conclusions}
In \cite{SO97} Sopena conjectures that every orientation of a connected cubic graph admits a $7$-colouring. 
Though we have not settled this conjecture we have improved the previous best known upper bound for the chromatic number of this family of graphs.
For every pair of oriented connected properly subcubic graphs $G_1$ and $G_2$, there exists an oriented connected properly subcubic graph $G$ that contains $G_1$ and $G_2$ as subgraphs.
Therefore there exists a tournament $T$ such that $T$ is a universal target for the family of oriented connected properly subcubic graphs and the order of $T$ is exactly the oriented chromatic number of this family. 
Finding such a $T$ on seven vertices would provide a much shorter proof for Theorem~\ref{orient:9colour}.
Further, if this $T$ has more than seven vertices, it would show the conjecture to be false. 
Computer search for such a tournament has eliminated $126$ of the $456$ non-isomorphic tournaments on seven vertices by observing that such a $T$ must contain every oriented  properly subcubic clique as a subgraph.

The method used for cubic graphs in Section \ref{cubic} could quite reasonably be applied for the case of graphs with maximum degree four. The tournament $QR_{11}$ could conceivably play the role of $QR_7$ as it has Property $P_{3,1}$ and the transitivity properties of $QR_7$.
Such a method could yield an improved bound of $13$ for the oriented chromatic number of quartic graphs. The oriented clique in Figure \ref{7clique} suggests that such a bound would be close to best possible.

\bibliographystyle{abbrv}
\bibliography{references}

\begin{thebibliography}{10}

\bibitem{BDS17}
J.~Bensmail, C.~Duffy, and S.~Sen.
\newblock Analogues of cliques for {$(m, n)$}-colored mixed graphs.
\newblock {\em Graphs and Combinatorics}, 33(4):735--750, 2017.

\bibitem{BEH81}
A.~Blass, G.~Exoo, and F.~Harary.
\newblock Paley {Graphs} satisfy all first-order {Adjacency} {Axioms}.
\newblock {\em Journal of Graph Theory}, 5:435--439, 1981.

\bibitem{B09}
A.~Bonato.
\newblock The {Search} for $n$-e.c. {Graphs}.
\newblock {\em Contributions to Discrete Mathematics}, 4:40--53, 2009.

\bibitem{BC06}
A.~Bonato and K.~Cameron.
\newblock On an {Adjacency} {Property} of almost all {Tournaments}.
\newblock {\em Discrete Mathematics}, 206:2327--2335, 2006.

\bibitem{DS14}
J.~Dybizba{\'n}ski and A.~Szepietowski.
\newblock The oriented chromatic number of {H}alin graphs.
\newblock {\em Information Processing Letters}, 114(1):45--49, 2014.

\bibitem{FRR03}
G.~Fertin, A.~Raspaud, and A.~Roychowdhury.
\newblock On the {Oriented} {Chromatic} {Number} of {Grids}.
\newblock {\em Inform. Proc. Letters}, 85:261--266, 2003.

\bibitem{GS71}
R.~Graham and J.~Spencer.
\newblock A constructive solution to a tournament problem.
\newblock {\em Canadian Mathematical Bulletin}, 14(1):45--48, 1971.

\bibitem{Gr83}
B.~Gr{\"{u}}nbaum.
\newblock Acyclic {Colorings} of {Planar} {Graphs}.
\newblock {\em Israel Journal of Mathematics}, 14:390--408, 1974.

\bibitem{KM04}
W.~Klostermeyer and G.~MacGillivray.
\newblock Analogues of {Cliques} for {O}riented {C}oloring.
\newblock {\em Discussiones Mathematicae Graph Theory}, 24(3):373--387, 2004.

\bibitem{KSZ97}
A.~Kostochka, E.~Sopena, and X.~Zhu.
\newblock Acyclic and {Oriented} {Chromatic} numbers of {Graphs}.
\newblock {\em Journal of Graph Theory}, 24(4):331--340, 1997.

\bibitem{PS06}
A.~Pinlou and {\'E}.~Sopena.
\newblock Oriented vertex and arc colorings of outerplanar graphs.
\newblock {\em Information Processing Letters}, 100(3):97--104, 2006.

\bibitem{RASO94}
A.~Raspaud and E.~Sopena.
\newblock {Good} and {Semi}-{Strong} {Colorings} of {Oriented} {Graphs}.
\newblock {\em Information Processing Letters}, 51:171--174, 1994.

\bibitem{S14}
S.~Sen.
\newblock {\em A {Contribution} to the {Theory} of {Graph} {Homomorphisms} and
  {Colourings}}.
\newblock PhD thesis, University of Bordeaux, 2014.

\bibitem{SO97}
E.~Sopena.
\newblock The {Chromatic} {Number} of {Oriented} {Graphs}.
\newblock {\em Journal of Graph Theory}, 25:191--205, 1997.

\bibitem{SO15}
E.~Sopena.
\newblock Homomorphisms and {Colourings} of {Oriented} {Graphs}: An updated
  survey.
\newblock {\em Discrete Mathematics}, 339(7):1993--2005, 2016.

\bibitem{SV96}
E.~Sopena and L.~Vignal.
\newblock A note on the {Oriented} {Chromatic} number of {Graphs} with
  {Maximum} {Degree} {Three}.
\newblock Technical report, LaBRI, University of Bordeaux I, 1996.

\bibitem{Wo07}
D.~R. Wood.
\newblock On the oriented chromatic number of dense graphs.
\newblock {\em Contributions to Discrete Mathematics}, 2(2), 2007.

\end{thebibliography}

\end{document}